\documentclass{iopart}
\def\eps{\varepsilon}
\def\fii{\varphi}
\def\d{ {\rm d} }

\def\scri{{\cal I}}
\newcommand{\qed}{\hfill $\Box$ \medskip}
\usepackage{iopams}
\usepackage{amssymb}
\usepackage{amsthm}
\usepackage{amsfonts}
\usepackage{fontenc}
\usepackage{graphicx}
\newtheorem{theo}{Theorem}[section]

\newtheorem{lemma}[theo]{Lemma}

\newtheorem{col}[theo]{Corollary}

\newcounter{mnotecount}[section]
\renewcommand{\themnotecount}{\thesection.\arabic{mnotecount}}

\newcommand{\mnote}[1]
{\protect{\stepcounter{mnotecount}}$^{\mbox{\footnotesize
$
\bullet$\themnotecount}}$ \marginpar{
\raggedright\tiny\em
$\!\!\!\!\!\!\,\bullet$\themnotecount: #1} }

\newcommand{\NUMPARTS}[1]{\numparts \def\theequation{#1\arabic{eqnval}{\it \alph{equation}}}}
\newcommand{\ENDNUMPARTS}[1]{\endnumparts \def\theequation{#1\arabic{equation}}}

\begin{document}

\title[No periodic asymptotically flat solutions of the Einstein-Maxwell equations.]{On asymptotically flat solutions \\of Einstein's equations periodic in time\\ I. Vacuum and electrovacuum solutions}
\author{J Bi\v{c}\'ak$^{1,3}$, M Scholtz$^{1,3}$, P Tod$^{2,1}$}
~\\
\address{$^1$\,Institute of Theoretical Physics, Faculty of Mathematics and Physics, \\ Charles University, V Hole\v{s}ovi\v{c}k\'ach 2, 180 00 Prague 8, Czech Republic}~\\
\address{$^2$\, Mathematical Institute, Oxford OX1 3LB, UK}~\\
\address{$^3$\,Max Planck Institute for Gravitational Physics, Albert Einstein Institute, \\
Am M\"uhlenberg 1, 14476 Golm, Germany}
\eads{\mailto{bicak@mbox.troja.mff.cuni.cz}, \mailto{scholtzzz@gmail.com}, \mailto{paul.tod@sjc.ox.ac.uk}}
\begin{abstract}
By an argument similar to that of Gibbons and Stewart \cite{GS}, but
in a different coordinate system and less restrictive gauge, we show
that any weakly-asymptotically-simple, analytic vacuum or
electrovacuum solutions of the Einstein equations which are periodic
in time are necessarily stationary.
\end{abstract}

\section{Introduction}
The inspiral and coalescence of
binary black holes or neutron stars appears to be the most promising
source for the detectors of gravitational waves, so that there has
been much effort going into the development of numerical codes and
analytic approximation methods to find the corresponding solutions
of Einstein's equations. One of the recent approaches assumes the
existence of a helical Killing vector $k$ (see e.g.
\cite{Workshop}). The field is assumed stationary in a rotating
frame where $k$ generates time translations but $k$ becomes null at
the light cylinder and is spacelike outside. $k$ has the form $k =\partial_t + \omega \partial_\phi$ where $\partial_t$ is timelike
and $\partial_\phi$ is spacelike with circular orbits with parameter
length $2\pi$ (except where $\partial_\phi=0$); $\omega=$constant.
The space-time is not stationary but it is still periodic where $k$
is spacelike. Requiring the helical symmetry for a binary system
implies equal amounts of outgoing and incoming radiation so that the
space-time, containing energy radiated all times is not expected to
be asymptotically flat. A corresponding solution in Maxwell's theory
for two opposite point charges moving on circular orbits was
considered a long time ago by Schild \cite{Schild}. The properties
of the field were analyzed recently in the Newman-Penrose formalism
in \cite{JiBiSchm}. The rather complicated periodicity properties of
the solution became apparent as well as its asymptotic behaviour: at
$\scri^-$ the advanced fields exhibit the standard Bondi-type
expansion and peeling, whereas the retarded fields do decay with
$r\rightarrow\infty$ but in an oscillatory manner like $(\sin r)/r$.
Hence for the retarded plus advanced solution no radiation field is
asymptotically defined.
Naturally, one would like to go beyond the linearized theory. There
are special exact, time-dependent, solutions known, for example,
Szekeres's dust solution, which has in general no Killing vector,
which can be matched to an exterior Schwarzschild metric \cite{Bo}. One can construct oscillating
spherical shells of dust particles moving with the same angular
momentum, but in every tangential direction, or oscillating Einstein
clusters which are matched to the Schwarzschild space-time outside
\cite{Gair}. Can there be periodic solutions representing ``bound
states" of gravitational or electromagnetic waves so that the
radiation field at infinity vanishes and the Bondi mass remains
constant?

There have been various attempts to prove that, while solutions of
the vacuum Einstein equations can be genuinely periodic in a
suitable time-coordinate (so \emph{not} time-independent), these
solutions cannot be asymptotically flat. These started with
\cite{pap1} and \cite{pap2}, with a summary in English in
\cite{pap3}, and \cite{AT} and more recently was considered in
\cite{GS}.
The method in \cite{pap1} considers vacuum metrics which are
everywhere nonsingular, weak and asymptoti\-cally-flat and which can
be expanded in a series in some parameter, with the flat metric as
the first term in the series. Each term in the series is assumed to
be periodic in a fixed Minkowski time-coordinate and to satisfy the
de Donder gauge condition. The second and third terms, call them
$v_{ab}$ and $w_{ab}$ respectively, are expanded as Fourier series
in the background time-coordinate and the Einstein equations then
imply that $v_{ab}$ satisfies the source-free wave equation, and
$w_{ab}$ satisfies a wave equation whose source is a quadratic
expression in $v_{ab}$. Assuming that the solution for $v_{ab}$ is
everywhere regular, the author shows that there cannot be an
asymptotically-flat solution for $w_{ab}$ unless $v_{ab}$ vanishes.
Therefore the space-time is flat. In \cite{pap2}, a similar
calculation when $v_{ab}$ is regular only outside a certain radius
leads to the conclusion that $v_{ab}$ must be time-independent in
order to have asymptotically-flat $w_{ab}$, and the space-time is
stationary.
In \cite{AT} it was observed by integrating the Einstein
pseudotensor and matter energy-momentum tensor over a 4-dimensional
volume that ``the mean value of power radiated by a periodic,
asymptotically Minkowskian gravitational field is equal to zero''.
The question of existence of periodic fields was left open.
In \cite{GS} the authors used the spin-coefficient formalism (see
e.g. \cite{NP}, \cite{JS}) to study the system of conformal Einstein
equations of Friedrich \cite{HF}. A coordinate system is based on
two families of null hypersurfaces, incoming from past null infinity
$\scri^-$ and labelled by constant $v$ and outgoing near $\scri^-$
and labeled by constant $u$. The authors define periodicity as meaning periodic in $v$ in these coordinates and are then able to prove that, at $\scri^-$, the
$u$-derivatives of all orders of all components of the metric are
independent of $v$. They conclude that if the metric is analytic in
these coordinates, then it necessarily has a Killing vector, which
in these coordinates is $\partial_v$, at least in a neighbourhood of
$\scri^-$. Thus any analytic metric, periodic in their sense, has
such a Killing vector. While certainly correct, there is a problem
with this conclusion in that, by construction, the Killing vector is
null wherever it is defined, and reduces at $\scri^-$ to a constant
translation along the generators. These are strong conditions and in
fact no Killing vector in flat space has these properties (any null
Killing vector is necessarily a null translation, and a null
translation is zero along one generator of $\scri$)\footnote{For
example the null translation $\partial_t+\partial_z$ becomes
$2\cos^2(\theta/2)\,\partial_v$ on $\scri^-$, which vanishes at
$\theta=\pi$.}. Thus flat space is not periodic according to the
definition of \cite{GS} and nor is any of the familiar stationary,
asymptotically flat solutions, for example the Schwarzschild
solution, so that this definition of periodicity is `too strong'.

For convenience, we follow \cite{GS} in working at $\scri^-$ rather
than $\scri^+$, though this is trivial to switch, but we shall make
a weaker definition of `periodic in time' which will permit metrics
stationary near $\scri^-$ and indeed will allow only these for
analytic, asymptotically-flat vacuum or electrovac metrics. We
follow the method of \cite{GS} for both the vacuum and electrovac
field equations, deferring other cases to a second article, but in a
different coordinate and tetrad system. Our coordinate and tetrad system is similar to
the one used at $\scri$ in \cite{NP}, and to prove the existence of
a symmetry at the event horizon in \cite{he} and at a compact Cauchy
horizon in \cite{IM}. We also differ from \cite{GS} in the choice of conformal gauge. In
\cite{GS} the unphysical Ricci scalar is set to zero by a choice of
conformal factor obtained by solving a wave equation. However, the
solution of the characteristic IVP for this wave equation as posed
in \cite{GS} will not in general be periodic, so that the rescaled,
unphysical metric would not in general share the periodicity of the
physical metric - in fact, in the particular case of the
Reissner-Nordstr\"om solution this gauge choice is compatible with
periodicity only for zero mass, as we show in Appendix C.
Thus we assume that there is at least one conformal factor which is
periodic and then modify this choice in the course of the
calculation in order to simplify the spin coefficients. From this point on, our method is then essentially the
same as in \cite{GS}, though a little more complicated, and we
arrive at the same conclusion, but now with a Killing vector which
is time-like in the interior, at least near to $\scri^-$.
The condition of time-like periodicity which we impose is as
follows: a space-time is time-like periodic if there is a discrete
isometry taking any point of the physical space-time to a point in
its chronological future. To define time-like periodicity at
$\scri^-$ for an asymptotically-flat space-time, we require this
isometry to extend to an isometry of a neighbourhood of
${\cal{I}}^-$ which preserves the generators of ${\cal{I}}^-$.  In
particular, we require the existence of at least one $\Omega$ which
conformally compactifies the space-time and preserves the
periodicity. The isometry has to be a supertranslation \cite{JS},
\begin{eqnarray}
\label{st}v\rightarrow v+a(\theta,\phi),
\end{eqnarray}
\noindent on ${\cal{I}}^-$, in terms of the usual coordinates
$(v,\theta,\phi)$ on $\scri^-$ and we shall assume that $a\neq 0$.
(We could imagine allowing $a$ to vanish on some generators of
$\scri^-$, since as noted above periodicity along a null translation
in flat space would appear like this at $\scri^-$, but this would be
null-periodicity rather than time-like periodicity.) We could assume
further that $a$ is actually a positive constant but this turns out
not to be necessary, as we shall find that, for analytic
space-times, this assumption of periodicity necessarily leads to a
space-time metric with a Killing vector which, in coordinates to be
defined, is $\partial_v$ and is time-like near $\scri^-$. Our result
is
\begin{theo}\label{one}
A weakly-asymptotically simple, vacuum or electrovac, time-periodic space-time
which is analytic in a neighbourhood of $\scri^-$ in the coordinates
introduced below necessarily has a Killing vector which is time-like
in the interior and extends to a translation on $\scri^-$.
\end{theo}
\medskip
Thus there are no \emph{non-trivial} time-periodic solutions satisfying these
conditions, in the sense that they would necessarily be actually
time-independent if time-periodic.
In a later article, we shall prove the corresponding result for the
Einstein equations coupled to either a massless scalar field with
the usual energy-momentum tensor, or a solution of the
conformally-invariant wave equation with the energy-momentum tensor
from p125 of \cite{PR} (sometimes called the `new improved energy
momentum tensor').

The method of proof requires the assumption of analyticity. It was
shown in \cite{HF2} that there are vacuum solutions analytic near
$\scri^-$. However, one would like either to drop the assumption of
analyticity, for example following the lead of \cite{FWR} or
\cite{gal} with a similar problem, or to prove that it follows from
the assumptions of periodicity and asymptotic-flatness. It remains
to be seen in what circumstances this can be done since, as noted
above, there are non-analytic solutions with matter in periodic
motion and matched to a (static) Schwarzschild exterior.

While this work is primarily motivated by an interest in the
possibility or impossibility of helical motions, it is worth noting
the connection with the question of the inheritance of symmetry.
Recall that, for a solution of Einstein's field equations with
matter, the matter is said to inherit the symmetry of the metric if
any isometry of the metric is necessarily a symmetry of the matter.
There are explicit solutions of the Einstein-Maxwell equations known
for which an isometry of the metric is \emph{not} a symmetry of the
Maxwell field \cite{MV} but these solutions are not
asymptotically-flat. In \cite{SKH} some other references may be
found for explicit solutions with Maxwell fields which do not share
the symmetry  of the metric. The same will be true for some
Robinson-Trautman solutions with null electromagnetic field which
may depend on time though the metric is static  (see \cite{SKH}, \S
28.2) These solutions will very likely have wire singularities
extending to infinity. From Theorem \ref{one} noninheritance cannot
happen with asymptotically-flat, analytic solutions:
%
\begin{col}\label{two}
In any weakly-asymptotically simple, stationary electrovac space-time
which is analytic in a neighbourhood of $\scri^-$ in the coordinates
introduced below, the Maxwell field is also stationary.
\end{col}
One can raise the question of inheritance also for Einstein-scalar
field solutions but the answer is rather different: for a massive
(complex) Klein-Gordon field there do exist solutions, the so-called
`boson stars', for which the metric is spherically-symmetric,
asymptotically-flat and static but the scalar field has a phase
linear in time (see e.g. \cite{BW}); however these solutions are not
analytic at infinity and, by a scaling argument, such
solutions do not exist with massless scalar fields. In a later
article, we shall obtain this result as a corollary of the result
corresponding to Theorem \ref{one}. In that subsequent work we start
from the conformal Einstein field equations with a general
energy-momentum tensor as a source and specialize them to scalar
field cases.

In Section \ref{ConfEMEq} we analyze the conformal Einstein-Maxwell
equations. We first rewrite Maxwell's equations in the unphysical
space-time, then translate the physical Bianchi identities and
obtain differential equations for the unphysical Weyl spinor and
Ricci spinor. In Appendix A we summarize a number of quantities,
their relations and behaviour under conformal transformations in the
Newman-Penrose formalism (\cite{NP}); these are extensively used in
the main text and in Appendices B and C. In particular, all
conformal equations for the gravitational and electromagnetic field
analyzed in terms of spinors in Section \ref{ConfEMEq} are projected
on the spin basis (i.e. the null tetrad) and written down in the
Newman-Penrose formalism in Appendix B.
In Section 3 a suitable coordinate system and a convenient
Newman-Penrose null tetrad which gives special values to some of the
Newman-Penrose spin-coefficients are introduced in the neighbourhood
of $\scri^-$. As noted above, these differ from those used by the
authors of \cite{GS}, and we shall show the differences explicitly in Section 3. In Appendix C we demonstrate that in contrast
to \cite{GS} our choice of gauge admits simple static (i.e.
`periodic') space-times like flat space and the Reissner-Nordstr\"om
metric.
In Section 4 we follow \cite{GS} (although in a different conformal
gauge) and study the problem in the NP formalism in the unphysical
space-time, with data on $\scri^-$. Assuming
periodicity along $\scri^-$ we first discover that the only
possibility is the independence of all geometric quantities of an
affine parameter $v$ along $\scri^-$. By induction we then prove
that all derivatives of all geometric quantities, including the
physical metric components, in the direction into the physical
space-time must also be $v-$independent. The proof of Theorem
\ref{one} and Corollary \ref{two} then follows from the assumed
analyticity.

This paper arose from a collaboration after P. T. posted his work
\cite{PT} on the gr-qc arXive and J. B. informed him that he and his
PhD. student M. S. were already engaged in tackling the same problem
\cite{SchBi}.


\section{The conformal Einstein-Maxwell equations}\label{ConfEMEq}

We first introduce conformal equations for the gravitational and
electromagnetic field in the formalism of 2-component spinors. In
Appendix \nolinebreak B these equations are written down explicitly
after the projection on a spin basis, in the form employed in the
Newman-Penrose formalism. In the physical space-time, Maxwell's
equations without sources are simply\footnote{Spinor indices are
labelled by $A, A^\prime, B, B^\prime,...$ and have values $0,1$.
The metric has signature -2.} (see e.g. \cite{JS})
\begin{eqnarray}
\tilde{\nabla}^{A A^\prime}\;\tilde{\phi}_{A B} = 0.\label{max1}
\end{eqnarray}
\noindent They are conformally invariant if under conformal
rescaling the Maxwell spinor $\phi_{AB}$ tranforms with conformal
weight $1$,
\begin{eqnarray}
\tilde{\phi}_{AB} = \Omega\,\phi_{AB},\label{ctm}
\end{eqnarray}
when the convention used in this article for conformal rescaling is
$\tilde\epsilon_{AB}=\Omega^{-1}\epsilon_{AB}$.

\noindent From the transformation of the derivative operator (see
(\ref{CovDerTrans})), in the unphysical space-time equations
(\ref{max1}) become
\begin{eqnarray}
\nabla^{A A^\prime}\;\phi_{AB} = 0.\label{MaxEqs}
\end{eqnarray}

 The situation is more complicated in the case of the
gravitational field. The physical Bianchi identities read
\nopagebreak
\begin{eqnarray}
\tilde{\nabla}^D_{C^\prime} \tilde{\Psi}_{ABCD} &=& \tilde{\nabla}^{D^\prime}_{(C}\tilde{\Phi}_{AB)C^\prime D^\prime},
\end{eqnarray}
\noindent where $\tilde{\Psi}_{ABCD}$ and $\tilde{\Phi}_{ABC^\prime
D^\prime}$ are the Weyl and the Ricci spinor, respectively. Using
the rules for the conformal transformation of these spinors (eq.
(\ref{ConfTrans}) and (\ref{WeylConfTrans})) we find

\begin{eqnarray}
\fl \hspace{0.2cm}\Omega^2 \nabla^D_{C^\prime}\psi_{ABCD} = \Omega\;\nabla^{D^\prime}_{(C}\Phi_{AB)C^\prime D^\prime}\;+\;\left(\nabla^{D^\prime}_{(A}\Omega\right) \Phi_{BC)C^\prime D^\prime}\;+\;\nabla^{D^\prime}_{(C}\nabla_{A(C^\prime}\nabla_{D^\prime)B)}\Omega,\nonumber
\\\label{Bianki}
\end{eqnarray}

\noindent where $\psi_{ABCD}=\Omega^{-1}\Psi_{ABCD}$. These
equations are the physical Bianchi identities written in terms of
the quantities in the unphysical space-time. We simplify them by
employing Einstein's equations in the physical space-time,
\begin{eqnarray}
\tilde{\Phi}_{ABA^\prime B^\prime} &=  k\,\tilde{\phi}_{AB}\;\bar{\tilde{\phi}}_{A^\prime B^\prime}.\label{EM-Rov}
\end{eqnarray}
\noindent Here we used the fact that the physical scalar curvature
vanishes for the electromagnetic field; we put the constant factor
$k$ on the r.h.s. of (\ref{EM-Rov}) equal to 1 following the
convention of \cite{NP}, unlike, e.g. \cite{PR}.
\noindent From eqs. (\ref{ConfTrans}), (\ref{EM-Rov}) and (\ref{ctm}) we obtain
\begin{eqnarray}
\nabla_{A(A^\prime}\nabla_{B^\prime)B}\Omega &=  \Omega^3\;\phi_{AB}\;\bar{\phi}_{A^\prime B^\prime}\;\;-\;\;\Omega
\;\Phi_{ABA^\prime B^\prime}.\label{DDO}
\end{eqnarray}
\noindent Applying $\nabla_C^{D^\prime}$, symmetrizing and using
Maxwell's equations (\ref{MaxEqs}), we can express the term
containing the third derivative of $\Omega$ appearing in
(\ref{Bianki}) as follows:
\begin{eqnarray}
\hspace{0.2cm}\fl\nabla^{D^\prime}_{(C}\nabla_{A(C^\prime}\nabla_{D^\prime)B)}\Omega\;=\;\nonumber\\ \fl \hspace{0.3cm} 3\Omega^2\bar{\phi}_{C^\prime D^\prime}\phi_{(AB}\nabla^{D^\prime}_{C)}\Omega\;+\;\Omega^3\bar{\phi}_{C^\prime D^\prime}\nabla^{D^\prime}_{(C}\phi_{AB)}\;-\;\Omega \nabla^{D^\prime}_{(C}\Phi_{AB)C^\prime D^\prime}\;-\;(\nabla^{D^\prime}_{(C}\Omega)\Phi_{AB)C^\prime D^\prime}.\nonumber
\end{eqnarray}
\noindent Inserting this result into (\ref{Bianki}) we arrive at the
conformal Bianchi identities for the Einstein-Maxwell field
expressed in terms of the quantities in the unphysical space-time:
\begin{eqnarray}
\nabla^D_{A^\prime}\psi_{ABCD} &=  3\,\bar{\phi}_{A^\prime B^\prime}\,\phi_{(A B}\,\nabla_{C)}^{B^\prime}\Omega\;\;+\;\;\Omega\,\bar{\phi}_{A^\prime B^\prime}\,\nabla_{(C}^{B^\prime}\,\phi_{AB)}.\label{Biankiki}
\end{eqnarray}
\noindent Projecting these equations onto the spin basis we obtain
the set of the equations which are explicitly written down (using
the NP formalism) in Appendix B, see (\ref{BBB1})-(\ref{BBB8}).
Equations (\ref{Biankiki}) are differential equations for the
unphysical Weyl spinor. To obtain the equations for the Ricci spinor
we use the Bianchi identities valid for quantities in the unphysical
space-time:
\begin{eqnarray}
\nabla^D_{C^\prime} \Psi_{ABCD} &=  \nabla^{D^\prime}_{(C}\Phi_{AB) C^\prime D^\prime}.
\end{eqnarray}

\noindent Combining the last two equations, we get
\nopagebreak
\begin{eqnarray}
\fl \hspace{0.2cm}\nabla^{B^\prime}_{(C}\Phi_{AB) A^\prime B^\prime} = \psi_{ABCD}\,\nabla^D_{A^\prime}\Omega\;+\;3\,\Omega\,\bar{\phi}_{A^\prime B^\prime}\,\phi_{(AB}\,\nabla_{C)}^{B^\prime}\,\Omega\;+\;\Omega^2\,\bar{\phi}_{A^\prime B^\prime}\nabla_{(C}^{B^\prime}\,\phi_{AB)}.\nonumber\\\label{Biankiki2}
\end{eqnarray}

In the following we shall also need the expression for quantities
$\nabla_{AA^\prime}\nabla_{BB^\prime}\Omega$. Let us decompose
$\nabla_{AA^\prime}\nabla_{BB^\prime}\Omega$ into its symmetric and
antisymmetric parts,
\begin{eqnarray}
\nabla_{AA^\prime}\nabla_{BB^\prime}\Omega = \nabla_{A(A^\prime}\nabla_{B^\prime)B}\Omega\;+\;\frac{1}{2}\;\epsilon_{A^\prime B^\prime}\;\nabla_{A C^\prime}\nabla_B^{C^\prime}\Omega.\label{derivationA}
\end{eqnarray}
\noindent The first term on the r.h.s. is given in (\ref{DDO}), the second term can be decomposed again:
\begin{eqnarray}
\nabla_{A C^\prime}\nabla_B^{C^\prime}\Omega &=  \nabla_{C^\prime (A}\nabla_{B)}^{C^\prime}\Omega\;\;+\;\;\frac{1}{2}\;\epsilon_{AB}\;\Box\Omega.\label{devD}
\end{eqnarray}
\noindent Since the operator $\nabla_{C^\prime (A}\nabla_{B)}^{C^\prime}$ is just the commutator $\nabla_{[a}\nabla_{b]}$ contracted by $\epsilon^{A^\prime B^\prime}$, it annihilates scalar quantities. Using equations (\ref{DDO}), (\ref{derivationA}) and (\ref{devD}) we obtain
\begin{eqnarray}
\nabla_{AA^\prime}\nabla_{BB^\prime}\Omega&= \Omega^3\,\phi_{AB}\,\bar{\phi}_{A^\prime B^\prime}\;-\;\Omega\,\Phi_{AB A^\prime B^\prime}\;+\;\frac{1}{4}\,\epsilon_{A^\prime B^\prime}\,\epsilon_{AB}\,\Box\Omega.\nonumber\\\label{devE}
\end{eqnarray}

It will be convenient to introduce the quantity
\begin{eqnarray}
F&= \frac{1}{2}\;\Omega^{-1}\;(\nabla_{A A^\prime}\Omega)\;(\nabla^{A A^\prime}\Omega),\label{DefOfF}
\end{eqnarray}
\noindent which can be seen to be smooth in the unphysical
space-time from the rule for the conformal transformation of the
scalar curvature (\ref{ConfTrans}) in the form
\begin{eqnarray}
\Box\Omega &=  4\,\Omega\,\Lambda\;-\;4\,\Omega^{-1}\,\tilde{\Lambda}\;+\;4\,F,
\end{eqnarray}
\noindent since the physical scalar curvature $\tilde{\Lambda}=0$
for the electromagnetic field. From equation (\ref{devE}) we now
obtain the following expressions for the second derivatives of
$\Omega$:
\begin{eqnarray}
\fl \nabla_{AA^\prime}\nabla_{BB^\prime}\Omega=\Omega^3\,\phi_{AB}\,\bar{\phi}_{A^\prime B^\prime}\;-\;\Omega\,\Phi_{AB A^\prime B^\prime}\;+\;\epsilon_{A^\prime B^\prime}\,\epsilon_{AB}\,(F\,+\,\Omega\,\Lambda).\nonumber\\\label{derE}
\end{eqnarray}

Finally we wish to derive expressions for $\nabla_{AA^\prime}F$.
Directly from the definition of the unphysical Riemann tensor and
from the decomposition (\ref{Riemann}), we have
\begin{eqnarray}
\fl\left(\nabla_{AA^\prime}\nabla_{BB^\prime} - \nabla_{BB^\prime}\nabla_{AA^\prime}\right)\nabla^{BB^\prime}\Omega &= - 2 \Phi_{ABA^\prime B^\prime} \nabla^{B B^\prime}\Omega + 6\,\Lambda\,\nabla_{AA^\prime}\,\Omega.\label{ContrRicci}
\end{eqnarray}
\noindent Employing Maxwell's equations (\ref{MaxEqs}) and the
contracted  Bianchi identities (\ref{SpinorBianchi}), we find that
equations (\ref{devE}) and (\ref{ContrRicci})  imply
\begin{eqnarray}
 \nabla_{AA^\prime}F &= \Omega^2 \,\phi_A^B \,\bar{\phi}_{A^\prime}^{B^\prime}\, \nabla_{BB^\prime}\Omega \,-\, \Phi_{ABA^\prime B^\prime}\nabla^{BB^\prime}\Omega + \Lambda\,\nabla_{AA^\prime}\Omega.\label{EqForF}
\end{eqnarray}
%

%
%

\section{Coordinates, tetrad and conformal gauge}\label{Kalibracia}

We assume that we have an analytic, time-periodic solution of the
Einstein-Maxwell equations and an analytic, time-periodic conformal
factor so that the unphysical metric with $\scri^-$ also has these
properties. We construct a convenient coordinate system and a Newman-Penrose null tetrad in
the neighborhood of $\scri^-$ (see Figs. \ref{NPcoords},
\ref{NPtetrad}). We stay in the unphysical space-time in order to
include $\scri^-$. Let ${\cal S}\subset\scri^-$ be a particular space-like 2-sphere. We
can introduce arbitrary coordinates $x^I, I=2,3$ on ${\cal S}$ and
propagate them along $\scri^-$ by the condition
\begin{eqnarray}
\nabla_{\dot{\gamma}}\,x^I &= 0,
\end{eqnarray}
\noindent where $\gamma=\gamma(v)$ is an affinely-parametrized  null
generator of $\scri^-$. We may set $v=0$ on ${\cal S}$. The triple
$(v, x^2, x^3)$ represents suitable coordinates on $\scri^-$. In
order to go into the interior of space-time we introduce the family
of null hypersurfaces ${\cal N}_v$ orthogonal to $\scri^-$ and
intersecting $\scri^-$ in the space-like cuts ${\cal S}_v$ of
constant $v$. Let $\gamma^\prime = \gamma^\prime(r)$ be the null
generators of the surface ${\cal N}_v$ labeled by $x^I$. Here, $r$
is the affine parameter which can be chosen so that $r=0$ on
$\scri^-$ and $g(\d v, \d r)=1$ at $\scri^-$. We propagate the
coordinates $v$ and $x^I$ onto ${\cal N}_v$ by conditions
\begin{eqnarray}
\nabla_{\dot{\gamma}^\prime}\,x^I\;=\;0,&\;\;&\nabla_{\dot{\gamma}^\prime}\,v\;=\;0.
\end{eqnarray}
\noindent We thus have established a coordinate chart
\begin{eqnarray}
x^\mu &=  (v,\;r,\; x^2,\;x^3)\;,\;\;\mu\;=\;0, 1, 2, 3\;,
\end{eqnarray}
\noindent in the neighbourhood of past null infinity.
\footnote{Components of tensors with respect to the basis induced by
these coordinates will be labelled by Greek letters $\mu, \nu, ...$.
Components with respect to an arbitrary tetrad will be labelled by
Latin letters $a, b, ...$ from the beginning of the alphabet.
Indices labelled by capital letters $I,J,..$ have values $2,3$. }

\begin{figure}[ht]
  \begin{center}
    \includegraphics[keepaspectratio,width=0.5\textwidth]{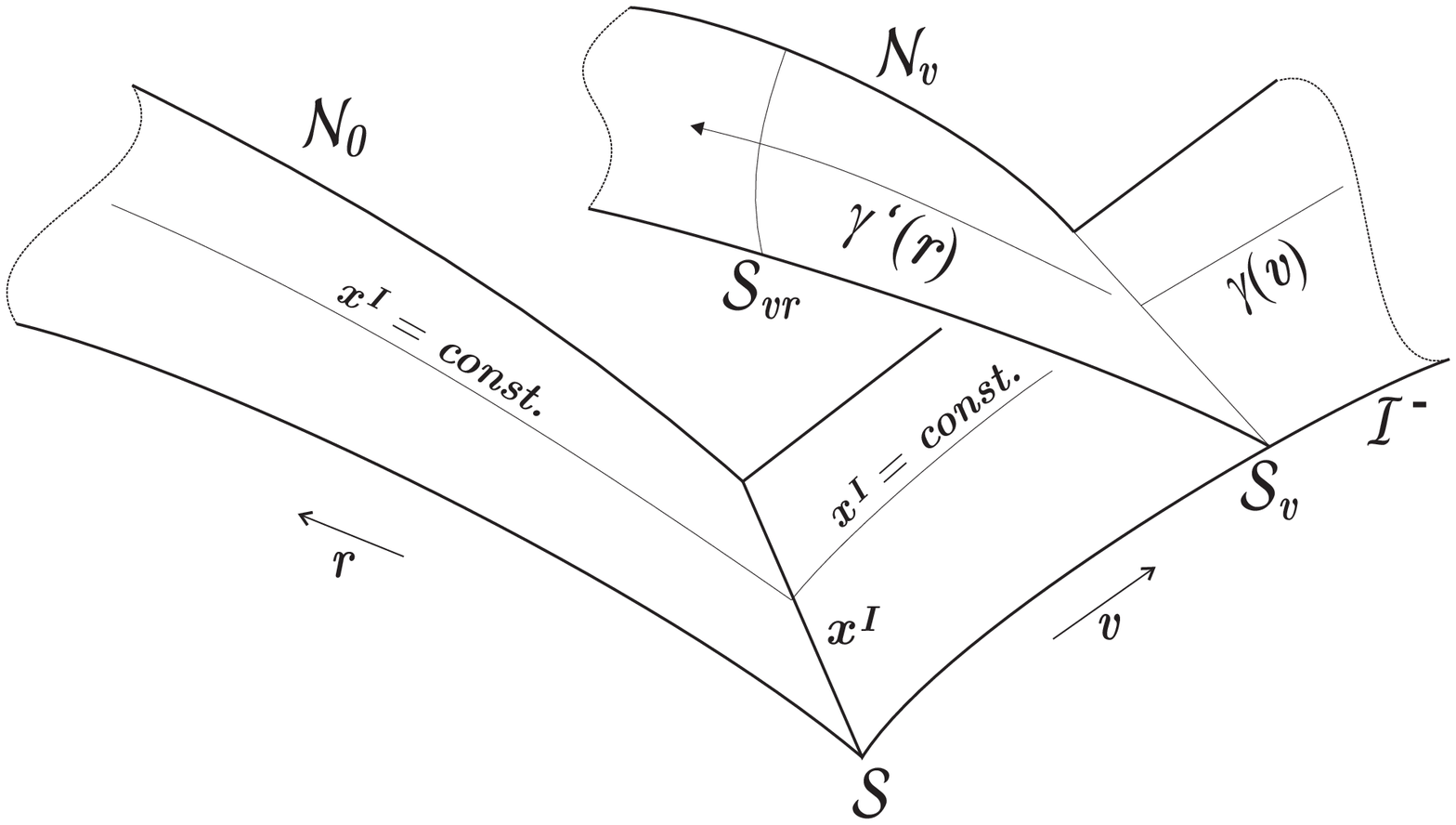}
    \caption{Construction of coordinate system. }
    \label{NPcoords}
  \end{center}
\end{figure}

Next we construct a suitable Newman-Penrose null tetrad. ${\cal
N}_v$ are null hypersurfaces $v = {\rm constant}$, therefore the
gradient of $v$ is both tangent and normal to ${\cal N}_v$; we
denote it by
\nopagebreak
\begin{eqnarray}
n_a &=  \nabla_a\;v\;.\label{ndef}
\end{eqnarray}
\noindent Since $n^a$ is tangent to $\gamma^\prime$ along which only $r$
varies,
\begin{eqnarray}
n&= \;\frac{\partial}{\partial r}\;\;.\label{Defn}
\end{eqnarray}
\noindent On each cut ${\cal{S}}_{vr}:v,r =$ constant there exists
exactly one null direction normal to ${\cal S}_{vr}$ not
proportional to $n^a$. We choose the vector field $l^a$ to be
parallel to this direction and normalize it by $n_a l^a=1$. It can
be written in the form \nopagebreak
\begin{eqnarray}
l &=  \frac{\partial}{\partial v}\;-\;H\,\frac{\partial}{\partial r}\;+\;C^I\,\frac{\partial}{\partial x^I}.\label{DefL}
\end{eqnarray}
\noindent On $\scri^-$ $l$ is tangent to the generators $\gamma(v)$,
so functions $H$ and $C^I$ vanish on $\scri^-$. The conformal gauge
can be chosen so that
\begin{eqnarray}
\frac{\partial \Omega}{\partial r} &=  1\;\;\;\;{\rm on}\;\;\scri^-.\label{DeltaOmega}
\end{eqnarray}

\begin{figure}[ht]
  \begin{center}
    \includegraphics[keepaspectratio,width=0.5\textwidth]{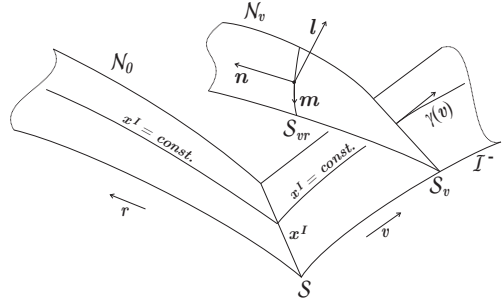}
    \caption{NP null tetrad. }
    \label{NPtetrad}
  \end{center}
\end{figure}

Let us now turn to the 2-spheres ${\cal S}_{vr}$ on which
$\partial_I$ are basis vectors. Since ${\cal S}_{vr}$ is a
space-like sphere, we choose, following standard procedure, a
complex vector $m$ and its complex conjugate $\bar{m}$, such that
\begin{eqnarray}
m^a\;m_a\;\;=\;\;0\;,&\;\; m^a\;\bar{m}_a\;\;=\;\;-1\;;
\end{eqnarray}
\noindent $m$ has the form
\begin{eqnarray}
m &=  P^I\;\frac{\partial}{\partial x^I}\;,\label{Defm}
\end{eqnarray}
\noindent where $P^2, P^3$ are complex functions. The coordinates
$x^I$ can be chosen to be the standard spherical coordinates,
$x^I=(\theta,\phi)$. Then the appropriate choice of the null vector
$m$ at $\scri^-$ is (see e.g.   \cite{JS})
\begin{eqnarray}
m\;\;=\;\; \frac{1}{\sqrt{2\,}}\,\left(\partial_\theta\;+\;\frac{i}{\sin\theta}\,\partial_\phi\right),&\;\;P^I\;\;=\;\;\frac{1}{\sqrt{2}}\,\left(1, \frac{i}{\sin\theta}\right).\label{SphrclCrds}
\end{eqnarray}
The vectors $m,\bar{m}$ are orthogonal to $l$ and $n$. The contravariant components of the tetrad read
\begin{eqnarray}
\eqalign{
l^\mu &=  (1, -H, C^2, C^3),\\
n^\mu &=  (0, 1, 0, 0),\label{tet2}\\
m^\mu &=  (0, 0, P^2, P^3).}
\end{eqnarray}
\noindent The contravariant components of the metric tensor are
given, regarding the relation $g^{\mu\,\nu}=2 l^{(\mu} n^{\nu)}-2
m^{(\mu}\bar{m}^{\nu)}$, by the matrix
\begin{eqnarray}
\fl g^{\mu\,\nu} &=  \left(
                   \begin{array}{cccc}
                     0 & 1 & 0 & 0 \\
                     1 & - 2 H & C^2 & C^3 \\
                     0 & C^2 & - 2 P^2\,\bar{P}^2 & - P^2\,\bar{P}^3- P^3\,\bar{P}^2 \\
                     0 & C^3 & - P^2\,\bar{P}^3- P^3\,\bar{P}^2 & -2 P^3 \, \bar{P}^3 \\
                   \end{array}
                 \right).                          \label{MetricTensor}
\end{eqnarray}
\noindent Using (\ref{tet2}) and the inverse of (\ref{MetricTensor}) we find the covariant components of the tetrad vectors:
\begin{eqnarray}
\eqalign{
l_\mu &=  (H, 1, 0, 0),\\
n_\mu &=  (1, 0, 0, 0),\\
m_\mu &=  (\omega, 0, R_2, R_3),}
\end{eqnarray}
\noindent where
\begin{eqnarray}
\eqalign{
R_2 \;=\;\frac{ P^3 }{P^2\,\bar{P}^3\,-\,P^3\,\bar{P}^2},&R_3 \;=\;\frac{ P^2 }{P^3\,\bar{P}^2\,-\,P^2\,\bar{P}^3},\\
\omega\;=\;-\,C^I\,R_I.}
\end{eqnarray}
\noindent The covariant components of the metric are
\begin{eqnarray}
\fl g_{\mu\,\nu} &=  \left(
                   \begin{array}{cccc}
                     2H-2\omega\bar{\omega} & 1 & -\omega \bar{R}_2-\bar{\omega}R_2  & -\omega \bar{R}_3-\bar{\omega}R_3 \\
                     1 & 0 & 0 & 0 \\
                     -\omega \bar{R}_2-\bar{\omega}R_2 & 0 & - 2 R_2 \bar{R}_2 & -R_3\bar{R}_2 - R_2\bar{R}_3 \\
                     -\omega \bar{R}_3-\bar{\omega}R_3 & 0 & -R_3\bar{R}_2 - R_2\bar{R}_3 & -2 R_3 \bar{R}_3 \\
                   \end{array}
                 \right).\label{CovariantMetric}
\end{eqnarray}
\noindent The vectors $l, n, m$ and $\bar{m}$ constitute the NP
tetrad. However, it is not unique since there is a rotation gauge
freedom $m\rightarrow e^{i\chi} m$ which will be used later.
Following the standard notation of the NP formalism (e.g.
\cite{NP}, \cite{JS}), we define the operators:
\begin{eqnarray}
D\;\;=\;\;l^a\;\nabla_a\;\;,\;\;\;\;\Delta\;\;=\;\;n^a\;\nabla_a\;\;,\;\;\;\;\delta\;\;=\;\;m^a\;\nabla_a\;.\label{DifOps}
\end{eqnarray}
\noindent We shall also employ the spin basis $(o^A,\iota^A)$ associated with the null tetrad,
\begin{eqnarray}
l^a\;\;=\;\;o^A\;\bar{o}^{A^\prime}\;,\;\;\;n^a\;\;=\;\;\iota^A\;\bar{\iota}^{A^\prime}\;,\;\;\;m^a\;\;=\;\;o^A\;\bar{\iota}^{A^\prime}\;,
\end{eqnarray}
\noindent normalized by $o_A\,\iota^A\,=\,1$.
Note that this coordinate and tetrad system has
some more gauge freedom associated with it. In particular we may
make another choice $\hat\Omega$ with $\hat\Omega=\Theta\Omega$
where $\Theta$ is also periodic and takes the value one at
$\scri^-$. Thus
\[
\tilde{g}_{ab}=\hat\Omega^{-2}\hat{g}_{ab}=\Omega^{-2}g_{ab},
\]
and so $\hat{g}_{ab}=\Theta^2g_{ab}$. We assume that $\Theta=1+f(v,r,x^I)$ with $f=O(r)$. This will change
the definition of the affine parameter $r$ , to $\hat{r}$ say, and then we must accompany the change of conformal factor with a null-rotation of the tetrad so that $\hat\delta$ is tangent to the sphere ${\cal{S}}_{v\hat{r}}$, thus

\begin{eqnarray}
\eqalign{
\hat{n}_a &= n_a,\\
\hat{m}_a &=\Theta(m_a+Zn_a),\\
\hat{l}_a &=\Theta^{2}(l_a+Z\bar{m}_a+\bar{Z}m_a+Z\bar{Z}n_a),\label{nullrot1}}
\end{eqnarray}

\noindent where $Z$, which parametrises the null rotation, is fixed by
requiring $\hat\delta \hat{r}=0$; the associated operators change
according to

\begin{eqnarray}
\eqalign{
\hat\Delta &= \Theta^{-2}\Delta,\\
\hat\delta &= \Theta^{-1}(\delta+Z\Delta),\\
\hat{D}    &= D+Z\bar\delta+\bar{Z}\delta+Z\bar{Z}\Delta.\label{nullrot2}}
\end{eqnarray}

With the coordinate $v$ common to both systems, we define $\hat{r}$
as the affine parameter with
\[\hat\Delta\hat{r}=\Theta^{-2}\Delta \hat{r}=1.\]
This can be integrated to give
\begin{eqnarray}\hat{r}&=\int^r_0\Theta^{2}dr=r+O(r^2),\label{rhat}\end{eqnarray}
and we need
\[0=\hat\delta\hat{r}=\Theta^{-1}(\delta\hat{r}+Z\Delta\hat{r}),\]
so that
\[Z=-\Theta^{-2}\delta\hat{r},\]
which can be calculated from (\ref{rhat}). Note that $Z=O(r^2)$. We
shall need to exploit this gauge freedom below.
Next we examine what special values some of the spin coefficients
take due to the above choice of the null tetrad (we calculate for
the unhatted system, but the same relations hold in the hatted
systems). Acting by the commutators
(\ref{commutators}) on the coordinate $v$, we find
\begin{eqnarray}
\gamma+\bar{\gamma}\;=\;\bar{\alpha}+\beta-\bar{\pi}\;=\;\nu\;=\;\mu-\bar{\mu} &=  0.\label{spingauge1}
\end{eqnarray}
\noindent Furthermore, commutators $[\delta,\Delta]r$ and $[\bar{\delta},\delta]r$ give
\begin{eqnarray}
\tau-\bar{\alpha}-\beta\;=\;\rho-\bar{\rho} &= 0.\label{spingauge2}
\end{eqnarray}
\nopagebreak
\noindent Applying the remaining commutators on the variables $v, r$
and $x^I$ leads to the ``frame equations", i.e.  the equations for
the metric functions $H, C^I$ and $P^I$:
\numparts
\begin{eqnarray}
\Delta\,H &=  -\,(\eps+\bar{\eps}),\label{FEq1}\\
\delta\,H &=  -\,\kappa,\label{FEq2}\\
\Delta\,C^I &=  -\,2\,\pi\,P^I\;-\;2\,\bar{\pi}\,\bar{P}^I,\label{FEq3}\\
\bar{\delta}\,P^I\;-\;\delta\,\bar{P}^I &=  (\alpha-\bar{\beta})\,P^I\;-\;(\bar{\alpha}-\beta)\,\bar{P}^I,\label{FEq4}\\
\Delta\,P^I &=  -\,(\mu - \gamma + \bar{\gamma})\,P^I\;-\;\bar{\lambda}\,\bar{P}^I,\label{FEq5}\\
\delta\,C^I\;-\;D\,P^I &=  -\,(\rho+\eps-\bar{\eps})\,P^I\;-\;\sigma\,\bar{P}^I.\label{FEq6}
\end{eqnarray}
\endnumparts
~\\
Since the generators $\gamma(v)$ of $\scri^-$ are affinely parametrized null geodesics, $D l^a = 0$ on $\scri^-$. Comparing this with
the general relation
\begin{eqnarray}
D l^a &=
(\varepsilon\;+\;\bar{\varepsilon})\,l^a\;\;-\;\;\bar{\kappa}\,m^a\;\;-\;\;\kappa\;\bar{m}^a\;,
\end{eqnarray}
\noindent  we see that
\begin{eqnarray}
\eps+\bar{\eps}\;=\;\kappa&= 0\;\;\;\;{\rm on}\;\;\scri^-.\label{gaugeeps0}
\end{eqnarray}

Next we wish to show that the freedom in choosing the basis
$(m,\bar{m})$ of the space tangential to ${\cal S}_{vr}$ allows us
to set $\gamma=0$. From the definition of $\gamma$ (eq.
\ref{SpinCoeff}) we have $\gamma-\bar{\gamma} =
m^a\Delta\bar{m}_a$. Under the rotation through $\chi$,
\nopagebreak
\begin{eqnarray}
m^a\rightarrow e^{i\chi}m^a, \label{mspin}
\end{eqnarray}

\noindent the quantity $\gamma-\bar{\gamma}$ transforms according to

\begin{eqnarray}
\gamma-\bar{\gamma} &\rightarrow& \gamma-\bar{\gamma} \;+\;i\Delta\chi,
\end{eqnarray}
\noindent so by solving the equation
\begin{eqnarray}
\Delta\chi &=  i(\gamma-\bar{\gamma}),
\end{eqnarray}
\noindent and regarding (\ref{spingauge1}) we can set
\begin{eqnarray}
\gamma &=  0.\label{gaugegamma}
\end{eqnarray}
Because the $\Delta-$operator is the derivative with respect  to the
coordinate $r$, further rotation (\ref{mspin}) with an
$r-$independent function $\chi$ does not violate the equality
(\ref{gaugegamma}). The quantity $\eps-\bar{\eps}$ under the
rotation (\ref{mspin}) transforms according to
\begin{eqnarray}
\eps\;-\;\bar{\eps} &\rightarrow& \eps\;-\;\bar{\eps}\;+\;i D\chi.
\end{eqnarray}
\noindent Solving the equation
\nopagebreak
\begin{eqnarray}
D\chi &= i(\eps\;-\;\bar{\eps})
\end{eqnarray}
\noindent on $\scri^-$, where $r=0$, we set $\eps=\bar{\eps}$ which, together with (\ref{gaugeeps0}), implies
\begin{eqnarray}
\eps &=  0\;\;\;\;{\rm on}\;\;\scri^-.\label{gaugeeps1}
\end{eqnarray}
To end this section, we exploit the gauge
freedom (\ref{nullrot1}) and (\ref{nullrot2}) to achieve a further
simplification. From the commutator $[\hat\delta,\hat\Delta]$ (see (\ref{commutators}) with the values of the spin coefficients fixed above) we
calculate
\[\hat\mu=\Theta^{-2}(\mu+\Theta^{-1}\Delta\Theta),\]
so that we can set $\hat\mu=0$ by choosing
\[\Theta=\mathrm{exp}\left(-\int^r_0\mu \,\d r\right).\]
Having done this, we omit the hats.

In order to elucidate the differences between our choice of the coordinate system and the null tetrad and those used by Gibbons and Stewart, we conclude this section by giving the details of their construction. Instead of the affine parameter $r$ they use coordinate $u$, defined as follows. Let ${\cal S}_0'$ be a spacelike cut on $\scri^-$ and ${\cal N}^\prime$, ${\cal S}'_0\subset {\cal N}^\prime$, the null hypersurface  such that the null generators of ${\cal N}^\prime$ are orthogonal to ${\cal S}'_0$. Now, the real function $u$ on ${\cal N}^\prime$ is defined in such a way, that $u=0$ on ${\cal S}_0'$, and $u=$constant on spacelike two-surfaces ${\cal S}_u$. The cut ${\cal S}_u$ defines another null hypersurface ${\cal N}_u$ with null generators orthogonal to ${\cal S}_u$. The coordinate $u$ is obtained by setting $u=$constant on ${\cal N}_u$. Similarly, the family of null hypersurfaces ${\cal N}^\prime_v$ orthogonal to spacelike cuts ${\cal S}^\prime_v$ on $\scri^-$, with $v$ being the affine parameter along the null generators of $\scri^-$, is constructed. Coordinates $x^I$ are chosen freely on ${\cal S}'_0$ and propagated into the spacetime along ${\cal N}^\prime$ and ${\cal N}_u$. The functions $x^\mu = (u,v,x^2, x^3)$ constitute a coordinate system in the neighbourhood of $\scri^-$ but note that in these coordinates the vector field $\partial_v$ is null, which it is not in our coordinates.

The NP tetrad used in \cite{GS} consists of vectors $l$, tangent to ${\cal N}_u$, $n$, tangent to ${\cal N}^\prime_v$, and $m,\bar{m}$ spanning the tangent space of ${\cal S}'_0$ and propagated into the space-time. The coordinate expression of the tetrad reads (this should be compared with our expressions (\ref{Defn}), (\ref{DefL}) and (\ref{Defm}))

\begin{eqnarray}
l \;=\;Q\,\partial_v,\;\;\;\;n\;=\;\partial_u + C^I\,\partial_I,\;\;\;\;m\;=\;P^I\,\partial_I,\label{GSNPtetrad}
\end{eqnarray}

\noindent where $Q, C^I$ and $P^I$ are metric functions. In this tetrad, the following equation holds:

\[
\Delta n^a \;=\;-(\gamma+\bar{\gamma}) n^a.
\]

\noindent Therefore, the null generators of ${\cal N}^\prime_v$ are geodesics, but $u$ is not an affine parameter.

The periodicity of the spacetime is defined as the periodicity of all geometrical quantities in the variable $v$. It is shown in \cite{GS} that $K=\partial_v$ is the Killing vector of the metric and concluded that the spacetime is stationary. However, $K$ is {\it null} everywhere by construction as it is tangent to the null generators of ${\cal N}_u$, while the stationarity requires the timelike Killing vector. Thus, it is impossible to conclude that the spacetime is stationary from the fact that $K$ is the Killing vector. As was mentioned in the introduction, even the Minkowski spacetime does not possess the Killing vector which is everywhere null and tangent to $\scri^-$.

In the following, we use the coordinates and the tetrad introduced in the beginning of this section.We show that $K=\partial_v$ is a space-time Killing vector which is null on $\scri^-$, but {\it timelike} in its neighbourhood.

\section{Proof of the theorem}

Having chosen coordinates and tetrad and fixed special values of
some of the NP coefficients we now analyse all geometric quantities
assuming analyticity in the chosen coordinates and periodicity on
$\scri^-$ in $v$. Following \cite{GS} we introduce the notation

\begin{eqnarray}
\eqalign{
S_0\;=\;D\Omega\;,\;\;S_1\;=\;\delta\Omega\;,\;\;S_2\;=\;\Delta\Omega\;,\label{Notation}
\\F\;=\;\frac{1}{\Omega}\left(S_0\,S_2\;-\;S_1\,\bar{S}_1\right)\;,
\psi_n\;=\;\frac{\Psi_n}{\Omega}\;,\;\;n=0,1,2,3,4,}
\end{eqnarray}

\noindent where $\Psi_n$ are the NP components of the Weyl spinor
(see eq. (\ref{psicomps2})). In the case of asymptotically-flat
space-time they vanish on $\scri^-$, so assuming smoothness, the
$\psi_n$ are regular there. Tangential derivatives of the conformal
factor vanish on $\scri^-$, i.e.  $S_0=S_1=0$, and so, again by
smoothness, the quantity $F$ is regular on $\scri^-$. The remaining
component of $\nabla\Omega$ is $S_2$ which is 1 on $\scri^-$ (cf.
(\ref{DeltaOmega})), so that its tangential derivatives also vanish
on $\scri^-$. Equations (\ref{derE}) and (\ref{EqForF}) are explicitly written
down in the NP formalism in Appendix B as (\ref{PrjA1}) -
(\ref{PrjB3}). Equations (\ref{PrjA4})-(\ref{PrjA10}) show that on $\scri^-$

\numparts
\begin{eqnarray}
\sigma &= \;0,\label{IC1}\\
F &= \; 0,\label{IC2}\\
\rho &=  \;0\label{IC7}\\
\bar{\pi}&= \;0\;\;=\;\;\beta\;+\;\bar{\alpha} \;=\; \tau,\label{IC3}\\
\Delta S_0 &= \; 0,\label{IC4}\\
\Delta S_2 &=  \;0,\label{IC5}\\
\Delta S_1 &=  \;0.\label{IC6}
\end{eqnarray}
\endnumparts
~\\
\noindent Since $F=0$ on $\scri^-$, also the tangential derivatives
$D F$ and $\delta F$ vanish there. From equations (\ref{PrjB1}) and
(\ref{PrjB2}) we thus obtain
\begin{eqnarray}
\Phi_{00}\;\;=\;\;\Phi_{01} &=  0\;\;\;\;{\rm on}\;\;\scri^-.\label{phi0001}
\end{eqnarray}
\noindent The metric functions $P^I$ on $\scri^-$ are given by
(\ref{SphrclCrds}). Inserting this expression into the frame
equation (\ref{FEq4}) and using relation (\ref{IC3}) we find
\begin{eqnarray}
\alpha\;\;=\;\;-\,\beta &=  -\,\frac{1}{2\,\sqrt{2}}\,\cot\theta\;\;\;\;{\rm on}\;\;\scri^-\label{alphabeta}.
\end{eqnarray}
\noindent The Ricci identity (\ref{RI17}) now shows that
\nopagebreak
\begin{eqnarray}
\Lambda\;+\;\Phi_{11} &=  \frac{1}{2}\;\;\;\;{\rm on}\;\;\scri^-.\label{LambdaPhi11}
\end{eqnarray}

In order to discover the behaviour of the other relevant quantities
we shall take into account the properties of the Bondi mass. In a
general asymptotically flat electrovacuum space-time the total mass-energy at
$\scri^+$ is defined by the formula (see e.g.  \cite{ENP})

\begin{eqnarray}
M_B &=  -\;\frac{1}{2\sqrt{\pi}}\int\d S\;\left( \tilde{\Psi}_2^0\;+\;\tilde{\sigma}^0\,\dot{\bar{\tilde{\sigma}}}^0\right).
\end{eqnarray}

\noindent By the superscript $0$ we denote the leading term in the
asymptotic expansion of a quantity, superscripts $1,2,...$ then
denote higher-order terms, for example, $\tilde{\sigma}=\tilde{\sigma}^0\,\Omega^2 + \tilde{\sigma}^1
\Omega^3 +{\cal O}(\Omega^4)$.  The rate of decrease of the Bondi
mass is given by

\begin{eqnarray}
\dot{M}_B &=  -\;\frac{1}{2\sqrt{\pi}}\int\d S\;\left( \dot{\tilde{\sigma}}^0\,\dot{\bar{\tilde{\sigma}}}^0\;+\;\tilde{\phi}_2^0\,\bar{\tilde{\phi}}_2^0\right).
\end{eqnarray}

\noindent The quantities $\sigma$ and $\phi_i$, $i = 0,1,2$, are defined in (\ref{SpinCoeff}) and (\ref{PhiComps}). Following the "conversion table" between $\scri^+$ and
$\scri^-$ (see (\ref{Kudryjas})), we analogously define the Bondi
mass at $\scri^-$ by

\begin{eqnarray}
M_B &= -\; \frac{1}{2\sqrt{\pi}}\int\d S\;\left( \tilde{\Psi}_2^0\;+\;\tilde{\lambda}^0\,\dot{\bar{\tilde{\lambda}}}^0\right).\label{Heqinur}
\end{eqnarray}

\noindent Since radiation comes into the physical
space-time through $\scri^-$ but can't exit through it, the total
mass energy at $\scri^-$ cannot decrease. Its rate of change in
(advanced) time $v$ along $\scri^-$ is given by
\nopagebreak
\begin{eqnarray}
\dot{M}_B &=  \frac{1}{2\sqrt{\pi}}\int\d S\;\left( \dot{\tilde{\lambda}}^0\,\dot{\bar{\tilde{\lambda}}}^0\;+\;\tilde{\phi}_0^0\,\bar{\tilde{\phi}}_0^0\right).
\end{eqnarray}

\noindent Now we assume periodicity. But a non-decreasing periodic
function must be a constant. Hence, our assumption of periodicity of
the mass-energy at $\scri^-$ requires

\begin{eqnarray}
\dot{\tilde{\lambda}}^0\;=\;0,\;\; \tilde{\phi}_0^0\;=\;0\;.
\end{eqnarray}

\noindent The leading term in the asymptotic expansion of
$\tilde{\Psi}_0$ is then
$\tilde{\Psi}_0^0=\ddot{\tilde{\bar{\lambda}}}^0=0$. Regarding
equations (\ref{PhiTrans}) and (\ref{PsiOrd}) and putting
$\tilde{\Psi}_0^0=0$, we can write the asymptotic expansion of
$\Psi_0$ near $\scri^-$ as
\begin{eqnarray}
\Psi_0 &=  \Psi_0^1\;\Omega^2\;+\;{\cal O}(\Omega^3),\label{psi0}
\end{eqnarray}
\noindent or (cf. eq. (\ref{Notation}))
\begin{eqnarray}
\psi_0 &=  {\cal O}(\Omega).
\end{eqnarray}
\noindent Equation (\ref{psi0}) implies
\begin{eqnarray}
\Delta\Psi_0 &=  0\;\;\;\;{\rm on}\;\;\scri^-\;.\label{DeltaPsi0}
\end{eqnarray}
\noindent Similarly, eq. (\ref{phitrans}), where we put
$\tilde{\phi}_0^0=0$, implies $\phi_0\in{\cal O}(\Omega)$ and
\nopagebreak
\begin{eqnarray}
\Delta \phi_0 &=  \phi_0^1\;S_2\;\;\;\;{\rm on}\;\;\scri^-\;.
\end{eqnarray}

The geometrical quantities consist
of the tetrad components, which give the metric functions, the spin
coefficients and the components of the Weyl and the Ricci tensor on
$\scri^-$. Because of our assumption of the periodicity of
gravitational field, the geometrical quantities are all assumed to be periodic in the variable $v$ on
$\scri^-$. We do not assume the periodicity of the electromagnetic
field since this field may not have the same symmetries as the
gravitational field (this is the issue of inheritance which we shall
return to).
We have shown that the following spin coefficients vanish on
$\scri^-$ (and thus do not depend on $v$):

\begin{eqnarray}
\mu,\rho,\sigma,\kappa,\eps,\nu,\gamma,\pi,\tau.
\end{eqnarray}

\noindent The spin coefficients $\alpha$ and $\beta$ are
$v-$independent because of (\ref{alphabeta}). Now we wish to show
that also the last spin coefficient $\lambda$ is independent of $v$.
The Bianchi identity (\ref{BIB1}) together with (\ref{DeltaPsi0})
and (\ref{phi0001}) shows that
\begin{eqnarray}
D\Phi_{02} &=  0\;\;\;\;{\rm on}\;\;\scri^-.\label{DPhi02}
\end{eqnarray}
\noindent If we now apply $D$ to the Ricci identity (\ref{RI7}), we
get
\begin{eqnarray}
D^2\lambda &=  0 \;\;\;\;{\rm on}\;\;\scri^-.
\end{eqnarray}
\noindent The general solution of this equation on $\scri^-$ is
\begin{eqnarray}
\lambda &=  \lambda^{(0)}\;+\;v\,\lambda^{(1)},
\end{eqnarray}
\noindent where $\lambda^{(0)}$ and $\lambda^{(1)}$ are functions
independent of $v$. Since $\lambda$ is assumed to be periodic and a
polynomial in $v$ can be periodic only if it is constant, we get
$\lambda=\lambda^{(0)}$ and
\begin{eqnarray}
D\lambda &=  0\;\;\;\;{\rm on}\;\;\scri^-.\label{Dlambda}
\end{eqnarray}
\noindent (we borrow this style of argument from \cite{GS} where it
is used extensively). The Ricci identity (\ref{RI7}) then
implies
\begin{eqnarray}
\Phi_{02} &=  0\;\;\;\;{\rm on}\;\;\scri^-.\label{Phi02}
\end{eqnarray}
\noindent The Ricci identity (\ref{RI8}) on $\scri^-$
becomes
\begin{eqnarray}
\Lambda &=0,\label{Dmu=l}
\end{eqnarray}
and then by (\ref{LambdaPhi11}) $\Phi_{11}=1/2$ there. Now from
(\ref{BIA3}) and $D$ on (\ref{RI11}), $D\Phi_{12}$ and $D\Phi_{22}$
vanish at $\scri^-$.
\medskip
\noindent We collect these
results and some similar ones as a lemma:
\begin{lemma}
The following are zero on $\scri^-$
\[H,C^A,\rho, \sigma,\pi,\kappa,\epsilon,S_0, S_1,
F,\psi_{0},\Phi_{00},\Phi_{01},\Phi_{02},\phi_0,\Lambda,\]
\[DP^A,D\alpha, D\beta,DS_2,D\lambda,D\Phi_{11}, D\Phi_{12},D\Phi_{22},D\psi_{1},D\psi_{2}, D\psi_{3},D\psi_{4},D\phi_1,D\phi_2,\]
\[D\Delta S_0,D\Delta S_1,D\Delta S_2.\]
\end{lemma}
\begin{proof}
The first line is done already, as is the second line up to
$D\psi_1$, which comes from (\ref{BBB1}). From $D$ applied to
(\ref{BBB2})-(\ref{BBB4}) we obtain $D^2\psi_i=0$ whence by
periodicity $D\psi_i=0$ at $\scri^-$, in order for $i=2,3,4$ . The
same procedure applied to (\ref{MR1}), (\ref{MR2}) takes care of
$D\phi_1,D\phi_2$. Then the third line follows from $D$ applied to
(\ref{PrjA8})-(\ref{PrjA10}).
\end{proof}
Now we turn to the proof of the Theorem. We set up an induction with
the following inductive hypothesis:
\bigskip
\noindent\emph{Suppose inductively that $\partial_v\Delta^jQ = 0$ at
$\scri^-$ for $0\leq j\leq k$ with $Q$ one of
\begin{eqnarray}\label{cond2}
&H,C^I,P^I,\epsilon,\pi,\lambda,\beta,\alpha,\rho,\sigma,\kappa,
F,\psi_{i},\Phi_{ij},\phi_i,\Lambda \end{eqnarray}
and for $0\leq j\leq k+1$ with $Q=S_i$.}

\noindent This is easily seen by the Lemma to hold for $k=0$, so we
need to deduce it for $j=k+1$ from its truth for $j\leq k$. In this
calculation, we use the fact that $\partial_v=D$ at $\scri^-$, and
make extensive use of the commutators (\ref{commutators}).
Under the inductive hypothesis, the inductive step follows
~\\

\begin{itemize}
\item
for $H, C^I,P^I$ from (\ref{FEq1}), (\ref{FEq3}) and (\ref{FEq5});
\item for $\epsilon,\pi,\lambda, \beta,\alpha,\rho,\sigma,\kappa$,
respectively, from (\ref{RI6}), (\ref{RI9}), (\ref{RI10}),
(\ref{RI12}), (\ref{RI15}), (\ref{RI14}), (\ref{RI13}) and
(\ref{RI3});
\item for $F$ from
(\ref{PrjB3});
\item
for $\phi_0$ and $\phi_1$ directly from (\ref{MR3}) and (\ref{MR4})
respectively; for $\phi_2$, from (\ref{MR2}) we deduce at $\scri^-$
\[D^2\Delta^{k+1}\phi_2= 0,\]
and then periodicity implies
\[D\Delta^{k+1}\phi_2= 0;\]
\item
for $\psi_{i}, i=0,1,2,3$ from (\ref{BBB5})-(\ref{BBB8}); for
$\psi_{4}$, under the inductive hypothesis, we deduce at $\scri^-$
\[D^2\Delta^{k+1}\psi_{4}= 0\]
from (\ref{BBB4}) and then periodicity implies
\[D\Delta^{k+1}\psi_{4}=0;\]
\item
for $\Phi_{00},\Phi_{01},\Phi_{02}, \Phi_{12}$ from (\ref{BIA2}),
(\ref{BIB2}),
 (\ref{BIA4}) and (\ref{BIB4}) respectively, all with $\Psi_{n}=\Omega\psi_{n}$;
then for $\Lambda$, $\Phi_{11}$ and $\Phi_{22}$ we use (\ref{RI8}),
(\ref{BIC3}) and (\ref{RI11}).
\end{itemize}
~\\
This completes the inductive step for the first set of quantities
$Q$. For $Q=S_i$ we use $D\Delta^{k+1}$ applied to
(\ref{PrjA8})-(\ref{PrjA10}).
Thus $r$-derivatives of all orders of the quantities in
(\ref{cond2}), which includes the metric functions $H,C^I$ and
$P^I$, are independent of $v$. Now analyticity in $r$ forces these
functions to be independent of $v$. Therefore, by (\ref{tet2}), the
metric components are all independent of $v$ and so
$K:=\partial/\partial v$ is a Killing vector of the unphysical
metric. However, for any $j$,
\[\partial_v\Delta^j\Omega=\partial_v\Delta^{j-1}S_2,\]
at $\scri^-$ and the r.h.s. vanishes for all $j$. Thus, by
analyticity in $r$, $\Omega$ is also independent of $v$ and so $K$
is a Killing vector of the physical metric too.
The norm-squared of the Killing vector is
\[g(K,K)=2(H-\omega\bar\omega).\]
This is $O(r^2)$ at $\scri^-$ but there
\[\Delta^2g(K,K)=2\Delta^2H=-2\Delta(\eps+\bar\eps)=2\]
so that $K$ is null at $\scri^-$ but time-like just inside: the
metric is stationary. \qed

This completes the proof of the theorem. Note that we have shown that, under the assumption of periodicity of the space-time and
the electromagnetic field, both fields are necessarily time-independent(in fact we assume slightly less, namely that the space-time is periodic and that $\phi_2$ is periodic).  A
slightly different question is whether a stationary
asymptotically-flat gravitational field might be produced by an
electromagnetic field which is not itself stationary. The content of
the Corollary 1.2 is that the answer is no.

~\\

\noindent {\bf Proof of Corollary 1.2.}\;\;\;Starting from the
assumption that the metric admits $\partial_v$ as a Killing vector,
we want to show that this is also a symmetry of the Maxwell field.
We have
\begin{eqnarray}
\tilde{\Phi}_{ij} &=  \Omega^2\,\phi_i\,\bar{\phi}_j,
\end{eqnarray}
\noindent and $\partial_v \tilde{\Phi}_{ij}=0$ so that, for some
$\chi$ possibly depending on $v$, we
have
\begin{eqnarray}
\phi_i &=  e^{i\,\chi}\,\fii_i,
\end{eqnarray}

\noindent where $\fii_i$ is $v-$independent. From the Maxwell
equation (\ref{MR1}), with $\phi_0=0$ on $\scri^-$, we find $\phi_1
D \chi=0$ on $\scri^-$ so that $D\chi=0$ unless $\phi_1=0$ there. If
$\phi_1=0$ there, (\ref{MR2}) gives $D\chi=0$ unless $\phi_2=0$, so
we can conclude that $D\phi_i=0$ on $\scri^-$. Now we set up an
induction to show that $D\Delta^n\phi_i=0$ on $\scri^-$ for all
$n\in{\mathbb N}$ and $i=0,1,2$. The inductive hypothesis will be
\begin{eqnarray}
(\forall k \leq n)(\forall i\in\{0,1,2\})(D\Delta^k\phi_i=0\;\;{\rm
on}\;\;\scri^-).
\end{eqnarray}
\noindent Then by $D\Delta^{n}$ on (\ref{MR3}) and (\ref{MR4}) we
obtain this for $k=n+1$ and $i=0,1$. For $i=2$, $D\Delta^{n+1}$ on
(\ref{MR2}) gives
\begin{eqnarray}
D^2\Delta^{n+1}\phi_2 &=  0\;\;\;\;{\rm on}\;\;\scri^-,
\end{eqnarray}
\noindent which integrates to give $\Delta^{n+1}\phi_2=a v + b$.
This would contribute a $v-$dependent term to $\tilde{\Phi}_{22}$ at
${\cal O}(\Omega^{2n+4})$, a contradiction unless $a=0$. Then
$D\Delta^{n+1}\phi_2=0$ on $\scri^-$, which completes the induction.

By assumption, the Maxwell field is analytic and so has a convergent
power series in $r$ near to $\scri^-$ and we have shown that all
coefficients are $v-$independent. Since the spinor dyad is
Lie-dragged by the Killing vector, this proves that the Maxwell
field is too: in this situation the Maxwell field inherits the
symmetry.
\qed
~\\
\ack{P. T. gratefully acknowledges hospitality and financial
support from the Mittag-Leffler Institute, Djursholm, Sweden and the
Charles University, Prague, and useful discussions with Gary Gibbons
and John Stewart. J.B. also acknowledges the discussions with Gary
Gibbons and the partial support from the Grant GA CR 202/09/00772 of
the Czech Republic, of Grant No LC06014 and MSM0021620860 of the
Ministry of Education. The work of M. S. was supported by the Grant
GAUK no. 22708 of the Charles University, Czech Republic.}
\newpage
\section*{Appendix A: The Newman-Penrose formalism and conformal transformations in Einstein-Maxwell space-times}\label{APP-A}
 \setcounter{equation}{0}
\renewcommand{\theequation}{A\arabic{equation}}
\subsection*{A1. Gravitational field}
In the NP formalism, the spin coefficients are the Ricci rotation
coefficients with respect to a null tetrad $\{l,n,m\}$ with the
corresponding spin basis ${o_A,\iota_A}$; they encode the
connection. The twelve independent complex coefficients are defined
by (see e.g. \cite{JS}, \cite{NP} for details)

\begin{eqnarray}
\eqalign{
\fl \kappa = m^a D l_a = o^A D o_A, &
\fl \hspace{3cm} \tau = m^a \Delta l_a = o^A \Delta o_A,  \\
\fl
\sigma = m^a \delta l_a = o^A \delta o_A, &
\fl \hspace{3cm}  \rho = m^a \bar{\delta} l_a = o^A \bar{\delta} o_A, \\
~\\
\fl \varepsilon = \frac{1}{2}\left[
n^a D l_a - \bar{m}^a D m_a\right] = \iota^A D o_A, &\fl \hspace{3cm}  \beta = \frac{1}{2}\left[
n^a \delta l_a - \bar{m}^a \delta m_a\right] = \iota^A \delta o_A,  \\
\fl \gamma = \frac{1}{2}\left[
n^a \Delta l_a - \bar{m}^a \Delta m_a\right] = \iota^A \Delta o_A, &\fl \hspace{3cm}  \alpha = \frac{1}{2}\left[
n^a \bar{\delta} l_a - \bar{m}^a \bar{\delta} m_a\right] = \iota^A \bar{\delta} o_A, \\ \\
\fl\pi = n^a D \bar{m}_a = \iota^A D \iota_A, &
\fl \hspace{3cm} \nu = n^a \Delta \bar{m}_a = \iota^A \Delta \iota_A,
 \\
\fl \lambda = n^a \bar{\delta} \bar{m}_a = \iota^A \bar{\delta} \iota_A,
&
\fl \hspace{3cm} \mu = n^a \delta \bar{m}_a = \iota^A \delta \iota_A  ,
}
\label{SpinCoeff}
\end{eqnarray}
~\\

\noindent where $D=\nabla_l,\;\Delta=\nabla_n,\;\delta=\nabla_m$. Acting on a scalar, the operators $D,\Delta, \delta$ obey the  commutation relations:

\begin{eqnarray}
\eqalign{
D\delta \;-\; \delta D  &=  (\bar{\pi}-\bar{\alpha}-\beta)D - \kappa
\Delta + (\bar{\rho}-\bar{\varepsilon}+\varepsilon)\delta + \sigma
\bar{\delta},\\ \Delta D  \;-\; D \Delta  &=
(\gamma+\bar{\gamma})D +(\varepsilon + \bar{\varepsilon})\Delta -
(\bar{\tau}+\pi)\delta -
(\tau+\bar{\pi})\bar{\delta},\label{commutators} \\ \Delta \delta
\;-\; \delta \Delta &=  \bar{\nu}D + (\bar{\alpha}+\beta -
\tau)\Delta + (\gamma-\bar{\gamma}-\mu)\delta - \bar{\lambda}
\bar{\delta},\\ \delta\bar{\delta}\; -\; \bar{\delta} \delta
&=  (\mu-\bar{\mu})D + (\rho-\bar{\rho})\Delta +
(\bar{\alpha}-\beta)\bar{\delta} -
(\alpha-\bar{\beta})\delta.
}
\end{eqnarray}

\noindent The Riemann tensor can be decomposed as follows:
\nopagebreak
\begin{eqnarray}
\eqalign{
R_{abcd} &=  C_{abcd} \label{Riemann}\\
&+\;\Phi_{AB C^\prime D^\prime}\;\epsilon_{A^\prime B^\prime}\;\epsilon_{CD}\;\;+\;\;\bar{\Phi}_{A^\prime B^\prime C D}\;\epsilon_{AB}\;\epsilon_{C^\prime D^\prime} \\
&+\;\Lambda\;\left(\epsilon_{AC}\;\epsilon_{BD}\;+\;\epsilon_{BC}\;\epsilon_{AD}\right)\;\epsilon_{A^\prime B^\prime}\;\epsilon_{C^\prime D^\prime} \\
&+\;\Lambda\;\left(\epsilon_{A^\prime C^\prime}\;\epsilon_{B^\prime D^\prime}\;+\;\epsilon_{B^\prime C^\prime}\;\epsilon_{A^\prime D^\prime}\right)\;\epsilon_{AB}\;\epsilon_{CD}.}
\end{eqnarray}
\noindent The first part is the Weyl tensor whose spinor equivalent is the totally symmetric Weyl spinor $\Psi_{ABCD}$:
\nopagebreak
\begin{eqnarray}
C_{abcd} &=  \Psi_{ABCD}\;\epsilon_{A^\prime B^\prime}\;\epsilon_{C^\prime D^\prime}\;\;+\;\;\bar{\Psi}_{A^\prime B^\prime C^\prime D^\prime}\;\epsilon_{AB}\;\epsilon_{CD}.
\end{eqnarray}
\noindent The scalar $\Lambda$ is related to the scalar curvature $R$ by
\begin{eqnarray}
\Lambda &=  \frac{1}{24}\;R.
\end{eqnarray}
\noindent The symmetric Ricci spinor $\Phi_{AB C^\prime D^\prime}$ is equivalent
to the trace-free part of the Ricci tensor:
\begin{eqnarray}
R_{ab} &=  - \;2\;\Phi_{A B A^\prime B^\prime}\;\;+\;\;6\;\Lambda\;\epsilon_{AB}\;\epsilon_{A^\prime B^\prime}.\label{CervenyHranol}
\end{eqnarray}
\noindent The spinor equivalent of the Einstein tensor is
\begin{eqnarray}
G_{ab} &=  - \;2\;\Phi_{A B A^\prime B^\prime}\;\;-\;\;6\;\Lambda\;\epsilon_{AB}\;\epsilon_{A^\prime B^\prime},\label{EinstSpinor}
\end{eqnarray}
\noindent and the spinor equivalent of Einstein's equations is
\begin{eqnarray}
\Phi_{AB A^\prime B^\prime} &=  -\;3\,\Lambda\,\epsilon_{AB}\,\epsilon_{A^\prime B^\prime}\;+\;4\,\pi\,T_{A B A^\prime B^\prime}.\label{EinstEQS}
\end{eqnarray}
\noindent Taking the symmetric part or contracting them with $\epsilon^{AB}\epsilon^{A^\prime B^\prime}$, respectively, we obtain two equations, equivalent to (\ref{EinstEQS}):
\begin{eqnarray}
\eqalign{
\Phi_{A B A^\prime B^\prime} \;=\; 4\,\pi\,T_{(A B) (A^\prime B^\prime)},\label{EinstEQSSpinor}\\
3\;\Lambda\;=\; \pi\,T_{A\;\;\;A^\prime}^{\;\;A\;\;\;A^\prime}.}
\end{eqnarray}
\noindent The five complex components of the Weyl spinor are
\begin{eqnarray}
\eqalign{
\Psi_0 = C_{abcd}l^a m^b l^c m^d &=\Psi_{ABCD}\,o^A o^B o^C o^D,\label{psicomps2}\\
\Psi_1 = C_{abcd}l^a n^b l^c m^d &=\Psi_{ABCD}\,o^A o^B o^C\iota^D,\\
\Psi_2 = C_{abcd}l^a m^b\bar{m}^c n^d &=\Psi_{ABCD}\,o^A o^B \iota^C \iota^D,\\
\Psi_3 = C_{abcd}l^a n^b\bar{m}^c n^d &=\Psi_{ABCD}\,o^A \iota^B \iota^C \iota^D,\\
\Psi_4 = C_{abcd}\bar{m}^a n^b\bar{m}^cn^d &=\Psi_{ABCD}\,\iota^A \iota^B \iota^C \iota^D.\\
}
\end{eqnarray}
\noindent The traceless Ricci tensor has the following components (3
real and 3 complex): \nopagebreak
\begin{eqnarray}\eqalign{
\Phi_{00} = - \frac{1}{2} R_{ab} l^a l^b  &=  \Phi_{ABA^\prime B^\prime} o^A o^B \bar{o}^{A^\prime} \bar{o}^{B^\prime},
\label{RicciComps}\\
\Phi_{01} = - \frac{1}{2} R_{ab} l^a m^b  &=  \Phi_{ABA^\prime B^\prime} o^A o^B \bar{o}^{A^\prime}\bar{\iota}^{B^\prime},\\
\Phi_{02} = - \frac{1}{2} R_{ab} m^a m^b  &=  \Phi_{ABA^\prime B^\prime} o^A o^B \bar{\iota}^{A^\prime}\bar{\iota}^{B^\prime},\\
\Phi_{11} = - \frac{1}{4} R_{ab} \left(l^a n^b + m^a \bar{m}^b\right)  &=  \Phi_{ABA^\prime B^\prime} o^A \iota^B\bar{o}^{A^\prime} \bar{\iota}^{B^\prime},\\
\Phi_{12} = - \frac{1}{2} R_{ab} n^a m^b  &=  \Phi_{ABA^\prime B^\prime} o^A \iota^B \bar{\iota}^{A^\prime}\bar{\iota}^{B^\prime},\\
\Phi_{22} = - \frac{1}{2} R_{ab} n^a n^b  &=  \Phi_{ABA^\prime B^\prime} \iota^A \iota^B \bar{\iota}^{A^\prime}\bar{\iota}^{B^\prime}.
}
\end{eqnarray}
\noindent The three remaining components can be obtained via the
condition $\Phi_{ij}=\bar{\Phi}_{ji}$.
Under the conformal rescaling $g_{ab}=\Omega^2\tilde{g}_{ab}$ the
covariant derivative acting on a 2-component spinor transforms as
\begin{eqnarray}
\tilde{\nabla}_{AA^\prime}\xi_B &=
\nabla_{AA^\prime}\xi_B\;+\;\Omega^{-1}\;\xi_A\;\nabla_{BA^\prime}\Omega.\label{CovDerTrans}
\end{eqnarray}
\noindent The NP quantities also transform. To find relations
between the physical and unphysical quantities we have to transform
the null tetrad. We wish to keep $n_a=\tilde{n}_a=\partial_av$ so
the correct choice is
\begin{eqnarray}
\hspace{-1cm}\eqalign{\label{BasisCT}\begin{array}{llll}
o^A\;=\;\tilde{o}^A\;,&\iota^A\;=
\;\Omega^{-1}\;\tilde{\iota}^A\;,&o_A\;=\;\Omega\;\tilde{o}_A\;,&\iota_A\;=\;\tilde{\iota}_A\;,
\\
l^a\;=\;\tilde{l}^a\;,&n^a\;=\;\Omega^{-2}\tilde{n}^a,&m^a\;=\;\Omega^{-1}\,\tilde{m}^a\;,&\bar{m}^a\;=\;\Omega^{-1}\,\bar{\tilde{m}}^a\;
\\
l_a\;=\;\Omega^2
\tilde{l}_a\;,&n_a\;=\;\tilde{n}_a\;,&m_a\;=\;\Omega\,\tilde{m}_a\;,&\bar{m}_a\;=\;\Omega\,\bar{\tilde{m}}_a,
\end{array}}
\end{eqnarray}
from which the transformation of the spin-coefficients can be
found.

The geometrical meaning of the spin coefficients depends
on the choice of the null tetrad. With our choices, the vector $l$
is pointing into $\scri^+$, while $n$ is tangent to $\scri^+$. On
$\scri^-$ the role of these vectors is interchanged, $n$ is pointing
from $\scri^-$ and $l$ is tangent to it. To convert quantities from
$\scri^+$ to $\scri^-$  we have only to switch the spinors $o^A$ and
$\iota^A$ (and adjust some signs). The
correspondence between the quantities on $\scri^+$ and $\scri^-$ is
given in the following table:
\begin{eqnarray}
\begin{array}{ll}
  \kappa \;\leftrightarrow\;\nu, & \tau \;\leftrightarrow\;\pi, \\
  \sigma \;\leftrightarrow\;\lambda, & \rho \;\leftrightarrow\;\mu, \\
  \varepsilon \;\leftrightarrow\;\gamma, & \alpha \;\leftrightarrow\;\beta, \\
  \Psi_{n}\;\leftrightarrow\;\Psi_{4-n}, & \Phi_{ij}\;\leftrightarrow\;\Phi_{(2-i)(2-j)}. \\
\end{array}\label{Kudryjas}
\end{eqnarray}
\noindent The scalar curvature and the Ricci spinor transform according to the
formulas
\nopagebreak
\begin{eqnarray}
\eqalign{
\tilde{R} = \Omega^2 R-6\Omega \Box \Omega+12
g^{ab}\left(\nabla_a\Omega\right)\left(\nabla_b\Omega\right),\label{ConfTrans}\\
\tilde{\Phi}_{A B A^\prime B^\prime}=\Phi_{A B A^\prime
B^\prime}+\Omega^{-1}\nabla_{A
(A^\prime}\nabla_{B^\prime) B}\Omega,}
\end{eqnarray}
\noindent the NP components of the Weyl spinor  as
\nopagebreak
\begin{eqnarray}
\tilde{\Psi}_n\;\;=\;\;\Omega^{n}\;\Psi_n\;. \label{PhiTrans}
\end{eqnarray}
\noindent The Weyl spinor is conformally invariant with weight zero:
\begin{eqnarray}
\Psi_{ABCD}=\tilde{\Psi}_{ABCD}\;.\label{WeylConfTrans}
\end{eqnarray}
\noindent Because the physical Weyl spinor vanishes on $\scri^-$, so
does the unphysical one, and assuming smoothness is therefore ${\cal
O}(\Omega)$. Then we get
\begin{eqnarray}
\tilde{\Psi}_n&\in&{\cal O}(\Omega^{n+1})\;.\label{PsiOrd}
\end{eqnarray}
\noindent The Ricci identities can be written in the spinor form as follows:
\begin{eqnarray}
\eqalign{
\nabla_{A^\prime (A}\nabla_{B)}^{A^\prime}\xi_C = \Psi_{ABCD}\xi^D-2\Lambda\xi_{(A}\epsilon_{B)C},\label{RicciIDSSp}\\
\nabla_{A (A^\prime}\nabla_{B^\prime)}^A\xi_C = \Phi_{C D A^\prime B^\prime}\xi^D.}
\end{eqnarray}

\noindent Substituting the basis spinors $o_A$ and $\iota_A$ for $\xi_A$ and projecting the last equations onto the spin basis we obtain the Ricci identities in the NP-formalism:

\NUMPARTS{A}
\begin{eqnarray}
\fl  D \rho - \bar{\delta} \kappa =\rho ^2+\left(\epsilon +\bar{\epsilon
   }\right) \rho -\kappa  \left(3 \alpha +\bar{\beta }-\pi \right)-\tau
   \bar{\kappa }+\sigma  \bar{\sigma }+\Phi_{00},\label{RI1}\\
\fl  D\sigma-\delta\kappa = (\rho+\bar{\rho}+3\eps-\bar{\eps})\sigma - (\tau-\bar{\pi}+\bar{\alpha}+3\beta)\kappa+\Psi_0,\label{RI2}\\
\fl  D\tau-\Delta\kappa = \rho(\tau+\bar{\pi})+\sigma(\bar{\tau}+\pi)+(\eps-\bar{\eps})\tau -(3\gamma+\bar{\gamma})\kappa+\Psi_1+\Phi_{01},\label{RI3}\\
\fl  D\alpha-\bar{\delta}\eps = (\rho +\bar{\eps}-2\eps)\alpha+\beta\bar{\sigma}-\bar{\beta}\eps - \kappa \lambda - \bar{\kappa}\gamma + (\eps+\rho)\pi + \Phi_{10},\label{RI4}\\
\fl  D\beta-\delta\eps = (\alpha+\pi)\sigma + (\bar{\rho}-\bar{\eps})\beta-(\mu+\gamma)\kappa-(\bar{\alpha}-\bar{\pi})\eps + \Psi_1,\label{RI5}\\
\fl  D\gamma-\Delta\eps = (\tau+\bar{\pi})\alpha + (\bar{\tau}+\pi)\beta - (\eps+\bar{\eps})\gamma - (\gamma + \bar{\gamma})\eps + \tau \pi - \nu \kappa\nonumber\\
 +\; \Psi_2 - \Lambda + \Phi_{11},\label{RI6}\\
\fl  D\lambda-\bar{\delta}\pi = (\rho - 3\eps+\bar{\eps})\lambda + \bar{\sigma}\mu + (\pi+\alpha-\bar{\beta})\pi - \nu\bar{\kappa}+\Phi_{20},\label{RI7}\\
\fl  D\mu-\delta\pi = (\bar{\rho}-\eps-\bar{\eps})\mu+\sigma\lambda+ (\bar{\pi}-\bar{\alpha}+\beta)\pi - \nu \kappa + \Psi_2 + 2 \Lambda,\label{RI8}\\
\fl  D\nu-\Delta\pi = (\pi+\bar{\tau})\mu+(\bar{\pi}+\tau)\lambda+(\gamma-\bar{\gamma})\pi - (3\eps+\bar{\eps})\nu+\Psi_3+\Phi_{21},\label{RI9}\\
\fl  \Delta\lambda-\bar{\delta}\nu = -(\mu+\bar{\mu}+3\gamma-\bar{\gamma})\lambda+(3\alpha+\bar{\beta}+\pi-\bar{\tau})\nu-\Psi_4,\label{RI10}\\
\fl  \Delta\mu-\delta\nu = -(\mu+\gamma+\bar{\gamma})\mu-\lambda\bar{\lambda}+\bar{\nu}\pi+(\bar{\alpha}+3\beta-\tau)\nu-\Phi_{22},\label{RI11}\\
\fl  \Delta\beta-\delta\gamma = (\bar{\alpha}+\beta-\tau)\gamma - \mu \tau + \sigma \nu + \eps \bar{\nu} + (\gamma-\bar{\gamma}-\mu)\beta - \alpha\bar{\lambda}-\Phi_{12},\label{RI12}\\
\fl  \Delta\sigma-\delta\tau = -(\mu-3\gamma+\bar{\gamma})\sigma - \bar{\lambda}\rho - (\tau + \beta - \bar{\alpha})\tau + \kappa \bar{\nu}-\Phi_{02},\label{RI13}\\
\fl  \Delta\rho-\bar{\delta}\tau = (\gamma+\bar{\gamma}-\bar{\mu})\rho - \sigma \lambda + (\bar{\beta}-\alpha-\bar{\tau})\tau + \nu \kappa - \Psi_2 - 2 \Lambda,\label{RI14}\\
\fl  \Delta\alpha-\bar{\delta}\gamma = (\rho+\eps)\nu - (\tau+\beta)\lambda + (\bar{\gamma}-\bar{\mu})\alpha + (\bar{\beta}-\bar{\tau})\gamma - \Psi_3,\label{RI15}\\
\fl  \delta\rho-\bar{\delta}\sigma = (\bar{\alpha}+\beta)\rho - (3\alpha-\bar{\beta})\sigma+(\rho-\bar{\rho})\tau+(\mu-\bar{\mu})\kappa -\Psi_1 + \Phi_{01},\label{RI16}\\
\fl  \delta\alpha-\bar{\delta}\beta = \mu\rho-\lambda\sigma + \alpha\bar{\alpha}+\beta\bar{\beta}-2\alpha\beta + (\rho-\bar{\rho})\gamma + (\mu-\bar{\mu})\eps  -  \Psi_2 + \Lambda + \Phi_{11},\label{RI17}\\
\fl  \delta\lambda-\bar{\delta}\mu = (\rho-\bar{\rho})\nu + (\mu-\bar{\mu})\pi + (\alpha+\bar{\beta})\mu+(\bar{\alpha}-3\beta)\lambda-\Psi_3 + \Phi_{21}.\label{RI18}
\end{eqnarray}
\ENDNUMPARTS{A}

\noindent The spinor form of the Bianchi identities is

\begin{eqnarray}
\nabla^D_{B^\prime} \Psi_{ABCD}&= \nabla_A^{A^\prime}\Phi_{BC A^\prime B^\prime}\;+\;\epsilon_{C(A}\,\nabla_{B)B^\prime}\Lambda\;-\;\frac{3}{2}\,\epsilon_{AB}\,\nabla_{CB^\prime}\Lambda.\nonumber\\
\label{SpinorBianchi}
\end{eqnarray}

\noindent Projecting these equations onto the spin basis leads to the Bianchi identities in the NP formalism:

\nopagebreak
\NUMPARTS{A}
\begin{eqnarray}\eqalign{{\fl}
D\Psi_1-\bar{\delta}\Psi_0-D\Phi_{01}+\delta\Phi_{00} =(\pi - 4 \alpha) \Psi_0+2(2\rho+\varepsilon)\Psi_1-3\kappa\Psi_2+2\kappa\Phi_{11}\\
\
 -\;(\bar{\pi}-2\bar{\alpha}-2\beta)\Phi_{00}-2\sigma\Phi_{10}-
2(\bar{\rho}+\varepsilon)\Phi_{01}+\bar{\kappa}\Phi_{02},\label{BIA1}}\\
\eqalign{
\fl D\Psi_2-\bar{\delta}\Psi_1+\Delta\Phi_{00}-\bar{\delta}\Phi_{01}+2D\Lambda=-\lambda\Psi_0 + 2 (\pi-\alpha)\Psi_1+3\rho \Psi_2-2\kappa\Psi_3\\
 +2\rho\Phi_{11}+\bar{\sigma}\Phi_{02}+\;(2\gamma+2\bar{\gamma}-\bar{\mu})\Phi_{00}-2(\alpha+\bar{\tau})\Phi_{01}-2\tau\Phi_{10},\label{BIA2}}\\
\eqalign{
\fl D\Psi_3-\bar{\delta}\Psi_2-D\Phi_{21}+\delta\Phi_{20}-2\bar{\delta}\Lambda = -2\lambda \Psi_1+3\pi\Psi_2 + 2 (\rho-\varepsilon)\Psi_3-\kappa\Psi_4\\
+2\mu\Phi_{10}-\;2\pi\Phi_{11}-(2\beta+\bar{\pi}-2\bar{\alpha})\Phi_{20}-2(\bar{\rho}-\varepsilon)\Phi_{21}+\bar{\kappa}\Phi_{22},\label{BIA3}}\\
\eqalign{
\fl D\Psi_4-\bar{\delta}\Psi_3+\Delta\Phi_{20}-\bar{\delta}\Phi_{21} =-3\lambda\Psi_2 +2(\alpha+2\pi)\Psi_3+(\rho-4\varepsilon)\Psi_4+2\nu\Phi_{10}\\
-2\lambda\Phi_{11}-\;(2\gamma-2\bar{\gamma}+\bar{\mu})\Phi_{20}-2(\bar{\tau}-\alpha)\Phi_{21}+\bar{\sigma}\Phi_{22},\label{BIA4}}
\end{eqnarray}
\ENDNUMPARTS{A}

\NUMPARTS{A}
\begin{eqnarray}
\eqalign{
\fl \Delta\Psi_0-\delta\Psi_1+D\Phi_{02}-\delta\Phi_{01}=(4\gamma-\mu)\Psi_0-2(2\tau+\beta)\Psi_1+3\sigma\Psi_2\\
+(\bar{\rho}+2\varepsilon-2\bar{\varepsilon})\Phi_{02}+\;2\sigma\Phi_{11}-2\kappa\Phi_{12}-\bar{\lambda}\Phi_{00}+2(\bar{\pi}-\beta)\Phi_{01},\label{BIB1}}\\
\eqalign{
\fl \Delta\Psi_1-\delta\Psi_2-\Delta\Phi_{01}+\bar{\delta}\Phi_{02}-2\delta\Lambda =\nu\Psi_0+2(\gamma-\mu)\Psi_1-3\tau\Psi_2+2\sigma\Psi_3\\
-\bar{\nu}\Phi_{00}+\;2(\bar{\mu}-\gamma)\Phi_{01}+(2\alpha+\bar{\tau}-2\bar{\beta})\Phi_{02}+2\tau\Phi_{11}-2\rho\Phi_{12},\label{BIB2}}\\
\eqalign{
\fl\Delta\Psi_2-\delta\Psi_3+D\Phi_{22}-\delta\Phi_{21}+2\Delta\Lambda=2\nu\Psi_1-3\mu\Psi_2+2(\beta-\tau)\Psi_3+\sigma\Psi_4\\
-2\mu\Phi_{11}-\bar{\lambda}\Phi_{20}+\;2\pi\Phi_{12}+2(\beta+\bar{\pi})\Phi_{21}+(\bar{\rho}-2\varepsilon-2\bar{\varepsilon})\Phi_{22},\label{BIB3}}\\
\eqalign{
\fl\Delta\Psi_3-\delta\Psi_4-\Delta\Phi_{21}+\bar{\delta}\Phi_{22}=3\nu\Psi_2-2(\gamma+2\mu)\Psi_3+(4\beta-\tau)\Psi_4-2\nu\Phi_{11}\\
-\bar{\nu}\Phi_{20}+\;2\lambda\Phi_{12}+2(\gamma+\bar{\mu})\Phi_{21}+(\bar{\tau}-2\bar{\beta}-2\alpha)\Phi_{22},\label{BIB4}}
\end{eqnarray}
\ENDNUMPARTS{A}

\NUMPARTS{A}
\begin{eqnarray}
\fl D\Phi_{11}-\delta\Phi_{10}+\Delta\Phi_{00}-\bar{\delta}\Phi_{01}+3D\Lambda = (2\gamma+2\bar{\gamma}-\mu-\bar{\mu})\Phi_{00}+(\pi-2\alpha-2\bar{\tau})\Phi_{01}\nonumber\\
\fl \hspace{0.5cm} +\;(\bar{\pi}-2\bar{\alpha}-2\tau)\Phi_{10}+2(\rho+\bar{\rho})\Phi_{11}+\bar{\sigma}\Phi_{02}
+\sigma\Phi_{20}-\bar{\kappa}\Phi_{12}-\kappa\Phi_{21},\label{BIC1}\\~\nonumber\\
\eqalign{
\fl D\Phi_{12}-\delta\Phi_{11}+\Delta\Phi_{01}-\bar{\delta}\Phi_{02}+3\delta\Lambda = (2\gamma-\mu-2\bar{\mu})\Phi_{01}+\bar{\nu}\Phi_{00}-\bar{\lambda}\Phi_{10}\\
\fl \hspace{0.5cm}+\;2(\bar{\pi}-\tau)\Phi_{11}+(\pi+2\bar{\beta}-2\alpha-\bar{\tau})\Phi_{02}
+(2\rho+\bar{\rho}-2\bar{\varepsilon})\Phi_{12}+\sigma\Phi_{21}-\kappa\Phi_{22},\label{BIC2}}\\~\nonumber\\
\fl D\Phi_{22}-\delta\Phi_{21}+\Delta\Phi_{11}-\bar{\delta}\Phi_{12}+3\Delta\Lambda = \nu \Phi_{01}+\bar{\nu}\Phi_{10}-2(\mu+\bar{\mu})\Phi_{11}-\lambda\Phi_{02}-\bar{\lambda}\Phi_{20}\nonumber\\
\fl \hspace{0.5cm}+\;(2\pi-\bar{\tau}+2\bar{\beta})\Phi_{12}+(2\beta-\tau+2\bar{\pi})\Phi_{21}
+\;(\rho+\bar{\rho}-2\varepsilon-2\bar{\varepsilon})\Phi_{22}.\label{BIC3}
\end{eqnarray}
\ENDNUMPARTS{A}

\subsection*{A2. Electromagnetic field}

For the description of an electromagnetic field we use the electromagnetic field
tensor $F_{ab}$ and its spinor equivalent $\phi_{AB}$:

\begin{eqnarray}
F_{ab} &=  \phi_{AB}\;\epsilon_{A^\prime B^\prime}\;\;+\;\;\bar{\phi}_{A^\prime B^\prime}\;\epsilon_{AB}.
\end{eqnarray}

\noindent The NP components of the Maxwell spinor are defined by

\nopagebreak

\begin{eqnarray}
\eqalign{
\phi_0 &=  F_{ab}\;l^a\;m^b\;\;=\;\;\phi_{AB}\;o^A\;o^B  \\
\phi_1 &=  \frac{1}{2}\;F_{ab}\;\left[l^a n^b \;-\; m^a \bar{m}^b\right]\;\;=\;\;\phi_{AB}\;o^A\;\iota^B \\
\phi_2 &=  F_{ab}\;\bar{m}^a\;n^b\;\;= \;\;\phi_{AB}\;\iota^A\;\iota^B
\label{PhiComps}}
\end{eqnarray}
\noindent The conformal transformation of these quantities is given by
\begin{eqnarray}
\tilde{\phi}_{AB}\;\;=\;\;\Omega\;\phi_{AB},&\;\;\tilde{\phi}_i\;\;=\;\;\Omega^{i+1}\;\phi_i\;.\label{phitrans}
\end{eqnarray}
\noindent Maxwell's equations without sources are equivalent to the
(conformally-invariant) spin-1 zero-rest-mass equation
\begin{eqnarray}
\nabla^{A}_{\; A^\prime}\,\phi_{AB} &=  0.
\end{eqnarray}

\noindent Projecting this onto the spin basis we obtain Maxwell's
equations in the NP formalism:
\NUMPARTS{A}
\begin{eqnarray}
D\phi_1\;-\;\bar{\delta}\phi_0 &=  (\pi - 2\alpha)\phi_0\;+\;2\rho\phi_1\;-\;\kappa \phi_2, \label{MR1}\\
D\phi_2\;-\;\bar{\delta}\phi_1 &=  -\lambda \phi_0 \;+\; 2\pi\phi_1\;+\;(\rho-2\varepsilon)\phi_2,\label{MR2}\\
\Delta\phi_0\;-\;\delta\phi_1 &=  (2\gamma -\mu)\phi_0\;-\;2\tau\phi_1\;+\;\sigma\phi_2,\label{MR3}\\
\Delta\phi_1\;-\;\delta\phi_2 &=  \nu \phi_0\;-\;2\mu\phi_1\;+\;(2\beta-\tau)\phi_2.\label{MR4}
\end{eqnarray}
\ENDNUMPARTS{A}

\section*{Appendix B: Conformal field equations}
\setcounter{equation}{0}
\renewcommand{\theequation}{B\arabic{equation}}

\subsection*{B1. Einstein-Maxwell fields}

\noindent The projections of the equation (\ref{derE}),
\begin{eqnarray}
\fl \nabla_{AA^\prime}\nabla_{BB^\prime}\Omega&= \Omega^3\,\phi_{AB}\,
\bar{\phi}_{A^\prime B^\prime}-\Omega\,\Phi_{AB A^\prime B^\prime}
+(F+\Omega\Lambda)\epsilon_{A^\prime B^\prime}\epsilon_{AB},
\end{eqnarray}

\noindent onto the null tetrad imply the following system of equations:

\NUMPARTS{B}
\begin{eqnarray}
\fl DS_0\;+\;(\eps+\bar{\eps})S_0\;+\;\bar{\kappa}S_1\;+\;\kappa\bar{S}_1 &=  \Omega^3 \phi_0 \bar{\phi}_0\;-\;\Omega\Phi_{00},\label{PrjA1}\\
\fl DS_1\;-\;\bar{\pi}S_0\;+\;(\bar{\eps}-\eps)S_1 \;+\;\kappa S_2 &=  \Omega^3 \phi_0 \bar{\phi}_1\;-\;\Omega\Phi_{01},\label{PrjA2}\\
\fl \delta S_0 - (\bar{\alpha}+\beta)S_0\;+\;\bar{\rho}S_1\;+\;\sigma \bar{S}_1 &=  \Omega^3 \phi_0 \bar{\phi}_1\;-\;\Omega\Phi_{01},\label{PrjA3}\\
\fl \delta S_1\;-\;\bar{\lambda}S_0\;+\;(\bar{\alpha}-\beta)S_1\;+\;\sigma S_2 &=  \Omega^3 \phi_0 \bar{\phi}_2\;-\;\Omega\Phi_{02},\label{PrjA4}\\
\fl DS_2\;-\;F\;-\;\Omega\,\Lambda\;-\;\pi S_1\;-\;\bar{\pi}\bar{S}_1\;+\;(\eps+\bar{\eps})S_2 &=  \Omega^3 \phi_1 \bar{\phi}_1\;-\;\Omega\Phi_{11},\label{PrjA5}\\
\fl \delta\bar{S}_1\;+\;F\;+\Omega\,\Lambda\;-\;\mu S_0\;+\;(\beta-\bar{\alpha})\bar{S}_1\;+\;\bar{\rho}S_2 &=  \Omega^3 \phi_1 \bar{\phi}_1\;-\;\Omega\Phi_{11},\label{PrjA6}\\
\fl \delta S_2\;-\;\mu S_1\;-\;\bar{\lambda}\bar{S}_1\;+\;(\bar{\alpha}+\beta)S_2 &=  \Omega^3 \phi_1 \bar{\phi}_2\;-\;\Omega\Phi_{12},\label{PrjA7}\\
\fl \Delta S_0 \;-\;F\;-\;\Omega\,\Lambda\;-\;(\gamma+\bar{\gamma})S_0\;+\;\bar{\tau}S_1\;+\;\tau \bar{S}_1&=  \Omega^3 \phi_1 \bar{\phi}_1\;-\;\Omega\Phi_{11},\label{PrjA8}\\
\fl \Delta S_1\;-\;\bar{\nu}S_0\;+\;(\bar{\gamma}-\gamma)S_1\;+\;\tau S_2&=  \Omega^3 \phi_1 \bar{\phi}_2\;-\;\Omega\Phi_{12},\label{PrjA9}\\
\fl \Delta S_2\;-\;\nu S_1\;-\;\bar{\nu} \bar{S}_1\;+\;(\gamma+\bar{\gamma})S_2&=  \Omega^3 \phi_2 \bar{\phi}_2\;-\;\Omega\Phi_{22}.\label{PrjA10}
\end{eqnarray}
\ENDNUMPARTS{B}

~\\
\noindent The projections of the equation (\ref{EqForF}),
\nopagebreak
\begin{eqnarray}
 \nabla_{AA^\prime}F = \Omega^2 \phi_A^B \bar{\phi}_{A^\prime}^{B^\prime}
 \nabla_{BB^\prime}\Omega - \Phi_{ABA^\prime B^\prime}\nabla^{BB^\prime}\Omega + \Lambda\,\nabla_{A A^\prime}\Omega,
\end{eqnarray}
\nopagebreak

\noindent give

\NUMPARTS{B}
\begin{eqnarray}
D F &= - S_0 \Phi_{11} + S_1 \Phi_{10}  +  \bar{S}_1 \Phi_{01} - S_2 \Phi_{00}\nonumber\\
 &+ \Omega^2\left[ S_0 \phi_1\,\bar{\phi}_1 - S_1 \phi_1\,\bar{\phi}_0  -  \bar{S}_1 \phi_0\,\bar{\phi}_1 + S_2 \phi_0\,\bar{\phi}_0\right] + \Lambda\,S_0,\label{PrjB1}\\
\delta F &= - S_0 \Phi_{12} + S_1 \Phi_{11}  +  \bar{S}_1 \Phi_{02} - S_2 \Phi_{01}\nonumber\\
 &+ \Omega^2\left[ S_0 \phi_1\,\bar{\phi}_2 - S_1 \phi_1\,\bar{\phi}_1  -  \bar{S}_1 \phi_0\,\bar{\phi}_2 + S_2 \phi_0\,\bar{\phi}_1\right] + \Lambda\,S_1,\label{PrjB2}\\
\Delta F &= - S_0 \Phi_{22} + S_1 \Phi_{21}  +  \bar{S}_1 \Phi_{12} - S_2 \Phi_{11}\nonumber\\
 &+ \Omega^2\left[ S_0 \phi_2\,\bar{\phi}_2 - S_1 \phi_2\,\bar{\phi}_1  -  \bar{S}_1 \phi_1\,\bar{\phi}_2 + S_2 \phi_1\,\bar{\phi}_1\right] + \Lambda\,S_2.\label{PrjB3}
\end{eqnarray}
\ENDNUMPARTS{B}

\noindent The conformal Bianchi identities (\ref{Biankiki}) for the Einstein-Maxwell field projected onto the spin basis imply the following system:\NUMPARTS{B}
\begin{eqnarray}
  \fl D\psi_1-\bar{\delta}\psi_0 = (\pi -4\alpha)\psi_0+2(\eps+2\rho)\psi_1-3\kappa\psi_2-3 S_1 \phi_0\bar{\phi}_0+3S_0 \phi_0\bar{\phi}_1\nonumber\\
+\Omega\,[ 2 \sigma \phi_1\bar{\phi}_0-2\beta\phi_0\bar{\phi}_0+2\eps\phi_0\bar{\phi}_1-2\kappa\phi_1\bar{\phi}_1+ \bar{\phi}_0 \delta\phi_0-\bar{\phi}_1D\phi_0],\label{BBB1}\\
\fl  D\psi_2-\bar{\delta}\psi_1 = -\lambda\psi_0+2(\pi-\alpha)\psi_1+2\rho\psi_2-2\kappa\psi_3-S_2 \phi_0\bar{\phi}_0-2 S_1 \phi_1\bar{\phi}_0\nonumber\\
    +2 S_0 \phi_1\bar{\phi}_1+\bar{S}_1 \phi_0\bar{\phi}_1+\frac{2}{3}\,\Omega\,[\bar{\phi}_0\delta\phi_1-\bar{\phi}_1D\phi_1-(\gamma+\mu)\phi_0\bar{\phi}_0\nonumber\\
   +\tau\phi_1\bar{\phi}_0+(\alpha+\pi)\phi_0\bar{\phi}_1+\sigma\phi_2\bar{\phi}_0-\rho\phi_1\bar{\phi}_1-\kappa\phi_2\bar{\phi}_1]+\frac{1}{3}\,\Omega\,\left[\bar{\phi}_0\Delta\phi_0-\bar{\phi}_1\bar{\delta}\phi_0\right],\label{BBB2}\\
\fl
  D\psi_3-\bar{\delta}\psi_2 = -2\lambda\psi_1+3\pi\psi_2+2(\rho-\eps)\psi_3-\kappa\psi_4-2 S_2 \phi_1\bar{\phi}_0-S_1 \phi_2\bar{\phi}_0+S_0\phi_2\bar{\phi}_1+2\bar{S}_1\phi_1\bar{\phi}_1\nonumber\\
   +\frac{2}{3}\,\Omega\,[-\nu\phi_0\bar{\phi}_0-\mu\phi_1\bar{\phi}_0+\lambda\phi_0\bar{\phi}_1+(\beta+\tau)\phi_2\bar{\phi}_0+\pi\phi_1\bar{\phi}_1-(\eps+\rho)\phi_2\bar{\phi}_1+\bar{\phi}_0\Delta\phi_1-\bar{\phi}_1\bar{\delta}\phi_1]\nonumber\\
   +\frac{1}{3}\,\Omega\,\left[ \bar{\phi}_0 \delta\phi_2-\bar{\phi}_1D\phi_2\right],\label{BBB3}\\
\fl
  D\psi_4-\bar{\delta}\psi_3 = - 3 \lambda \psi_2+2(\alpha+2\pi)\psi_3+(\rho-4\eps)\psi_4-3 S_2 \phi_2\bar{\phi}_0+3 \bar{S}_1 \phi_2\bar{\phi}_1\nonumber\\
   +\Omega\,\left[\bar{\phi}_0\Delta\phi_2-\bar{\phi}_1\bar{\delta}\phi_2-2\nu\phi_1\bar{\phi}_0+2\gamma\phi_2\bar{\phi}_0+2\lambda\phi_1\bar{\phi}_1-2\alpha\phi_2\bar{\phi}_1\right],\label{BBB4}\\
\fl
  \delta\psi_1-\Delta\psi_0 = (\mu-4\gamma)\psi_0 + 2(\beta+2\tau)\psi_1-3\sigma\psi_2-3S_1 \phi_0\bar{\phi}_1+3S_0\phi_0\bar{\phi}_2\nonumber\\
+\Omega\,[-2\beta\phi_0\bar{\phi}_1+2\sigma\phi_1\bar{\phi}_1+2\eps\phi_0\bar{\phi}_2-2\kappa\phi_1\bar{\phi}_2-\bar{\phi}_2D\phi_0+\bar{\phi}_1\delta\phi_0],\label{BBB5}\\
\fl  \delta\psi_2-\Delta\psi_1 =-\nu\psi_0+2(\mu-\gamma)\psi_1+3\tau\psi_2-2\sigma\psi_3-S_2 \phi_0\bar{\phi}_1-2 S_1 \phi_1\bar{\phi}_1+2 S_0 \phi_1\bar{\phi}_1+\bar{S}_1 \phi_0\bar{\phi}_2\nonumber\\
+\frac{2}{3}\,\Omega\,[-(\gamma+\mu)\phi_0\bar{\phi}_1+\tau\phi_1\bar{\phi}_1+\sigma\phi_2\bar{\phi}_1+(\pi+\alpha)\phi_0\bar{\phi}_2-\rho\phi_1\bar{\phi}_2-\kappa\phi_2\bar{\phi}_2\nonumber\\
+\bar{\phi}_1\delta\phi_1-\bar{\phi}_2 D\phi_1]+\frac{1}{3}\,\Omega\,\left[\bar{\phi}_1\Delta\phi_0-\bar{\phi}_2\bar{\delta}\phi_0\right],\label{BBB6}\\
\fl \delta\psi_3-\Delta\psi_2 = - 2\nu \psi_1+3\mu\psi_2+2(\tau-\beta)\psi_3 -\sigma \psi_4- 2 S_2 \phi_1\bar{\phi}_1 - S_1\phi_2\bar{\phi}_1+S_0\phi_2\bar{\phi}_2+2\bar{S}_1 \phi_1\bar{\phi}_2\nonumber\\
+\frac{2}{3}\,\Omega[-\nu\phi_0\bar{\phi}_1-\mu\phi_1\bar{\phi}_1+(\beta+\tau)\phi_2\bar{\phi}_1+\lambda\phi_0\bar{\phi}_2+\pi\phi_1\bar{\phi}_2-(\eps+\rho)\phi_2\bar{\phi}_2+\bar{\phi}_1\Delta\phi_1-\bar{\phi}_2\bar{\delta}\phi_1]\nonumber\\
+\frac{1}{3}\,\Omega\left[ \bar{\phi}_1\delta\phi_2-\bar{\phi}_2D\phi_2\right],\label{BBB7}\\
\fl  \delta\psi_4-\Delta\psi_3 = - 3\nu\psi_2 + 2 (2\gamma+2\mu)\psi_3 + (\tau-4\beta)\psi_4-3 S_2 \phi_2\bar{\phi}_1+3 \bar{S}_1 \phi_1\bar{\phi}_2\nonumber\\
+\Omega\left[ - 2 \nu \phi_1\bar{\phi}_1 + 2\gamma\phi_2\bar{\phi}_1+2\lambda\phi_1\bar{\phi}_2-2\alpha\phi_2\bar{\phi}_2+\bar{\phi}_1\Delta\phi_2-\bar{\phi}_2\bar{\delta}\phi_2\right].\label{BBB8}
\end{eqnarray}
\ENDNUMPARTS{B}

\section*{Appendix C: Reissner-Nordstr\"om space-time}
\setcounter{equation}{0}
\renewcommand{\theequation}{C\arabic{equation}}
To justify  our choice of gauge and show that the
choice made by \cite{GS} is too restrictive we shall show here how a
simple space-time, namely, the Reissner-Nordstr\"om solution,
appears in our gauge. The physical metric is
\begin{eqnarray}
\fl \d\tilde{s}^2 = \left(1\,-\,\frac{2 m}{\tilde{r}}\,+\,\frac{Q^2}{\tilde{r}^2}\right)\,\d t^2 - \left(1\,-\,\frac{2 m}{\tilde{r}}\,+\,\frac{Q^2}{\tilde{r}^2}\right)^{-1}\,\d \tilde{r}^2 - \tilde{r}^2\,\d \Sigma^2,
\end{eqnarray}
\noindent where $Q$ is the charge and $\d\Sigma^2=\d
\theta^2+\sin^2\theta \d \phi^2$. In the standard conformal
compactification of the Reissner-Nordstr\"om space-time one
introduces the ``tortoise coordinate" $r^*$ and the advanced time
$v$ by \nopagebreak
\begin{eqnarray}
\eqalign{
\d\tilde{r} = \left(1\;-\;\frac{2m}{\tilde{r}}\;+\;\frac{Q^2}{\tilde{r}^2}\right)\,\d r^*,\\
v=t\;+\;r^*.}
\end{eqnarray}
\noindent In these coordinates the physical metric acquires the form
\begin{eqnarray}
\d \tilde{s}^2 = \left(1 - \frac{2m}{\tilde{r}} + \frac{Q^2}{\tilde{r}^2}\right)\,\left(\d v^2 - 2\,\d v\,\d r^*\right) - \tilde{r}^2\,\d\Sigma^2.
\end{eqnarray}
\noindent We compactify it by defining the coordinate
\begin{eqnarray}
r = \tilde{r}^{-1}
\end{eqnarray}
\noindent and the conformal factor
\begin{eqnarray}
\Omega = r.
\end{eqnarray}
\noindent The unphysical metric then reads
\begin{eqnarray}\label{rn5}
\d s^2 = r^2\,\left(1-2\,m\,r+Q^2\,r^2\right)\,\d v^2\;+\;2\,\d
v\,\d r \;-\;\d\Sigma^2.
\end{eqnarray}

\noindent Comparing this with (\ref{MetricTensor})-(\ref{CovariantMetric}) we find the metric functions to be
\begin{eqnarray}
\eqalign{
H &= \frac{1}{2}\,r^2\;-\;m\,r^3\;+\;\frac{1}{2}\,Q^2\,r^4,\\
C^I &= 0,\\
P^2 &= \frac{1}{\sqrt{2}},\\
P^3 &= \frac{i}{\sqrt{2}\,\sin\theta}.}
\end{eqnarray}
\noindent From the metric the other geometrical quantities follow. The spin coefficients are all zero, except for
\begin{eqnarray}
\eqalign{
\eps &= -\,\frac{1}{2}\,r + \frac{3}{2}\,m\,r^2 - Q^2\,r^3,\\
\alpha &= -\,\beta = -\,\frac{1}{2\,\sqrt{2}}\,\cot\theta.}
\end{eqnarray}

\noindent The non-zero components of the Weyl and Ricci tensor read
\begin{eqnarray}
\eqalign{
\psi_2 &= m - Q^2\,r,\\
\Phi_{11} &= \frac{1}{2} - \frac{3}{2}\,m\,r + \frac{3}{2}\,Q^2\,r^2,\\
\Lambda &= \frac{1}{2}\,m\,r - \frac{1}{2}\,Q^2\,r^2.}
\end{eqnarray}
\noindent The electromagnetic 4-potential and corresponding electromagnetic tensor in these coordinates are
\begin{eqnarray}
\eqalign{A_\mu &= \left( Q\,r,0,0,0\right),\\
F_{\mu\,\nu} &= -\,Q\,\epsilon_{\mu\,\nu\,2\,3}.}
\end{eqnarray}
\noindent The only non-vanishing NP component of $F_{\mu\nu}$ is
\begin{eqnarray}
\phi_1 &= Q,
\end{eqnarray}
\noindent as one would expect. All these results are in accordance
with results obtained in the text.
On the other hand, the gauge condition $\Lambda=0$ everywhere,
imposed in \cite{GS}, leads to a periodic unphysical metric only if
$m=0$, i.e. flat space-time. This can be seen as follows: we need to
rescale the metric (\ref{rn5}) say to
\[\hat{g}_{ab}=\Theta^{-2}g_{ab}\]
so that, by (\ref{ConfTrans}),
\[\hat\Lambda=\Theta^{-2}(\Theta\Lambda+\frac14\Box\Theta)=0,\]
where the boundary conditions on $\Theta$ are that $\Theta=1$ on
$r=0$ and, say, $v=0$ (in order to preserve the conditions that
$\rho=0$ on $r=0$, $\mu=0$ on $v=0$ and $\Theta=1$ on $v=r=0$). With
the metric (\ref{rn5}) this wave equation on $\Theta$ becomes
\begin{eqnarray}\label{rn6}2\partial_v\partial_r\Theta&=\partial_r(A\partial_r\Theta)-L^2\Theta-2(mr-Qr^2)\Theta,\end{eqnarray}
with $A=r^2(1-2mr+Q^2r^2)$ and
\[L^2\Theta=\frac{1}{\sin\theta}\frac{\partial}{\partial\theta}
\left(\frac{\partial\Theta}{\partial\theta}\right)+\frac{1}{\sin^2\theta}\frac{\partial^2\Theta}{\partial\theta^2}\;.\]
Now from (\ref{rn6}) evaluated at $r=0$ we calculate
$\partial_v\partial_r\Theta=0$ so that $\partial_r\Theta$ is
constant on $\scri^-$, but it vanishes at $v=0$ so it is zero for
all $v$. Then from (\ref{rn6}) again at $\scri^-$,
\[\partial_v\partial^2_r\Theta=-m.\]
Thus $\Theta$ cannot be periodic in $v$ unless $m=0$, in which case
the physical metric is flat.
~\\
~\\

\end{document}